\documentclass[journal]{IEEEtranTIE}
\usepackage{graphics} % for pdf, bitmapped graphics files
\usepackage{epsfig} % for postscript graphics files
\usepackage{amsmath} % assumes amsmath package installed
\usepackage{amssymb}  % assumes amsmath package installed
\usepackage{amsthm}  % assumes amsmath package installed

\usepackage{color}
\usepackage{url}

\usepackage{algorithmic}
\usepackage{algorithm}
\usepackage{enumerate}
\usepackage{here}
\usepackage{flushend}%最終頁の高さを揃える
\usepackage{mathtools}
\usepackage{multirow} 
\usepackage{makecell}  % プリアンブル
\usepackage{cuted} 
\usepackage{stfloats} %長い数式を1段組で表示するときに使用した．
\usepackage{xcolor}

\usepackage{amsthm} % theorem-->IEEEtran では不要

\usepackage[nocompress]{cite}  % assumes amsmath package installed
\usepackage[]{hyperref}%%% Hyperref %%%

\hypersetup{
 hidelinks,
	colorlinks=false, % false: リンク色は付けない（枠で表示）
	citecolor=blue,
	linkcolor=red,
	urlcolor=orange,
	citebordercolor=green,
	linkbordercolor=red,
	urlbordercolor=cyan,
}

\makeatletter
\def\ps@plain{
  \def\@oddhead{\hfil\thepage}%
  \def\@evenhead{\thepage\hfil}%
  \def\@oddfoot{}%
  \def\@evenfoot{}%
}
\makeatother

\begin{document}
% \pagenumbering{arabic} % ←これが決定打（document直後に置く）
% \bstctlcite{BSTcontrolTIE}

\newtheorem{rem}{Remark}
\newtheorem{theorem}{Theorem}
\newtheorem{lemma}{Lemma}
\newtheorem{corollary}{Corollary}
\newtheorem{proposition}{Proposition}

\title{	Generalized bilinear Koopman realization from input-output data
for multi-step prediction with metaheuristic optimization of lifting function and its application to real-world industrial system
% \\ (Apr. 2021)
}

\author{
	\vskip 1em
	Shuichi Yahagi, \emph{Membership},
	Ansei Yonezawa, \emph{Membership},
    Heisei Yonezawa, \emph{Membership},\\
    Hiroki Seto, \emph{Membership},
	and Itsuro Kajiwara
	\thanks{
		Manuscript received Month xx, 2xxx; revised Month xx, xxxx; accepted Month x, xxxx.
		This work was supported in part by the ISUZU Advanced Engineering Center, Ltd.
		
		Shuichi Yahagi \emph{(Corresponding author)} is with Department of Mechanical Engineering, Tokyo City University, 1-28-1 Tamazutsumi, Setagaya-ku, Tokyo 158-8557, Japan (e-mail: yahagisi@tcu.ac.jp).

        Ansei Yonezawa is with Department of Mechanical Engineering, Kyushu University, 744 Motooka, Nishi-ku, Fukuoka 819-0395, Japan. 
        % (e-mail: ayonezawa@mech.kyushu-u.ac.jp).
        
        Heisei Yonezawa and Itsuro Kajiwara are with Division of Mechanical and Aerospace Engineering, Hokkaido University, N13, W8, Kita-ku, Sapporo, Hokkaido 060-8628, Japan. 
        % (e-mail: yonezawah@eng.hokudai.ac.jp; ikajiwara@eng.hokudai.ac.jp).
        
        Hiroki Seto are with 6th Research Department, ISUZU Advanced Engineering Center Ltd., 8 Tsutidana, Fujisawa-shi, Kanagawa 252-0881, Japan.
        %  (e-mail: hiroki\_seto@isuzu.com).  
	}
}

\maketitle
\thispagestyle{plain} % ←タイトル直後に入れると1ページ目にも番号が出る

\begin{abstract}
This paper introduces an input-output bilinear Koopman realization with an optimization algorithm of lifting functions. For nonlinear systems with inputs, Koopman-based modeling is effective because the Koopman operator enables a high-dimensional linear representation of nonlinear dynamics. However, traditional approaches face significant challenges in industrial applications. Measuring all system states is often impractical due to constraints on sensor installation. Moreover, the predictive performance of a Koopman model strongly depends on the choice of lifting functions, and their design typically requires substantial manual effort. In addition, although a linear time-invariant (LTI) Koopman model is the most commonly used model structure in the Koopman framework, such model exhibit limited predictive accuracy. To address these limitations, we propose an input-output bilinear Koopman modeling in which the design parameters of radial basis function (RBF)-based lifting functions are optimized using a global metaheuristic algorithm to improve long-term prediction performance. Consideration of the long-term prediction performance enhances the reliability of the resulting model. The proposed methodology is validated in simulations and experimental tests, with the airpath control system of a diesel engine as the plant to be modeled. This plant represents a challenging industrial application because it exhibits strong nonlinearities and coupled multi-input multi-output (MIMO) dynamics. These results demonstrate that the proposed input-output bilinear Koopman model significantly outperforms traditional linear Koopman models in predictive accuracy.
\end{abstract}


\markboth{}%
{}

\definecolor{limegreen}{rgb}{0.2, 0.8, 0.2}
\definecolor{forestgreen}{rgb}{0.13, 0.55, 0.13}
\definecolor{greenhtml}{rgb}{0.0, 0.5, 0.0}

\section{Introduction}%INTRODUCTION
\subsection{Motivation}
\IEEEPARstart{S}ystem identification is essential for enabling system prediction, control, and fault detection. Traditional system identification methods, such as Auto-Regressive with eXogenous (ARX) models and Numerical Algorithms for Subspace State Space System Identification (N4SID)~\cite{van_overschee_n4sid_1994}  have proven highly effective for linear systems. However, identifying models for highly nonlinear industrial systems remains challenging, which in turn complicates the control of complex industrial processes. To address this issue, data-driven approaches for nonlinear systems have been proposed, including the Koopman operator framework~\cite{korda_linear_2018,o_williams_kernel-based_2015,proctor_generalizing_2018}, deep learning methods such as recurrent neural networks (NNs) and long short-term memory (LSTM) networks~\cite{huang_lstm-mpc_2023}, nonlinear ARX (NARX) models~\cite{zhang_bayesian_2024}, and sparse identification of nonlinear dynamics (SINDy)~\cite{brunton_sparse_2016,bhattacharya_nonlinear_2022}.

The Koopman operator is formulated as a linear operator in an infinite-dimensional space by mapping the original state space to a higher-dimensional space through lifting functions~\cite{brunton_modern_2022,bevanda_koopman_2021}. Its ability to represent a nonlinear system as a linear model is a distinctive feature not found in other data-driven approaches. Although Koopman operator theory inherently deals with infinite-dimensional spaces, in practical computations, a finite-dimensional approximation of the Koopman operator can be identified from data using dynamic mode decomposition (DMD)~\cite{schmid_dynamic_2010,proctor_dynamic_2016} and extended DMD (EDMD)~\cite{williams_datadriven_2015,williams_extending_2016}. The resulting Koopman model can accurately capture the original nonlinear behavior by leveraging a high-dimensional linear structure. This enables the application of conventional linear control theory~\cite{korda_linear_2018,meng_online_2025,zhao_kalman-koopman_2024}. 
Despite these achievements, several challenges for application to industrial systems remain, including the selection of lifting functions that significantly impact prediction performance, the inherent limitations of linear model predictions, and  the unavailability of measurements for the full system states.

\subsection{Related work}
It has been noted that LTI Koopman models may not adequately capture the control-affine dynamics of nonlinear systems, which makes accurate prediction challenging~\cite{yu_autonomous_2022}. To address this limitation, bilinear Koopman realizations have been investigated. These models offer improved predictive accuracy compared to LTI Koopman models while remaining computationally more efficient than fully nonlinear Koopman approaches~\cite{bruder_advantages_2021}. Theoretical foundations of bilinear Koopman realizations have been explored in~\cite{bruder_advantages_2021,goswami_global_2017}. In addition, methods for computing bilinear lifting functions using deep NNs have been proposed~\cite{wang_deep_2024,zhao_deep_2024,abtahi_deep_2025}. 
The choice of the observation function that maps state variables to observable variables--commonly referred to as the lifting function--has a significant impact on the predictive performance of the Koopman model. According to~\cite{shi_koopman_2023}, previous approaches for designing lifting functions include mechanics-inspired selections, empirical selections such as monomials and polynomials, and RBF functions with randomly assigned centers. Another proposed approach employs highly contributory basis functions derived from SINDy modeling~\cite{zhang_reduced-order_2023,wang_improved_2023}. However, these methods often fail to ensure reproducibility or improve the predictive accuracy of model identification. To address these limitations, recent research has explored neural network-based approaches for lifting function design~\cite{lusch_deep_2018,yang_adaptive_2026}. Nevertheless, NN-based methods present challenges such as hyperparameter tuning, overfitting, model complexity, computational cost, and issues related to vanishing or exploding gradients. Furthermore, the orthogonality of NN-based lifting functions is not guaranteed. To date, meta-heuristic approaches such as particle swarm optimization (PSO) have not been investigated for lifting function design.

Previously, LTI and bilinear Koopman realizations typically assumed that all system states and inputs were measurable. However, in practical industrial applications, it is often infeasible to deploy sensors capable of capture every state variable. 
To address this limitation, a Koopman realization framework that relies solely on input-output data is required. This perspective is particularly important for industrial implementation. Moreover, prior research on Koopman operators using input-output data~\cite{korda_linear_2018,xiawen_parameters_2021,wang_improved_2023} has been limited, with most studies focusing exclusively on LTI Koopman formulations. Additionally, the selection of the arguments (i.e., embedded states) of the lifting function for controlled systems has not been thoroughly investigated.

Although modeling of diesel engine airpath systems has been widely investigated~\cite{moriyasu_diesel_2019,zhao_explicit_2014,huang_nitrogen_2023}, achieving accurate and robust models remains highly challenging.
These systems exhibit nonlinear and multivariable dynamics. 
The traditional approaches rely on physical modeling~\cite{zhao_explicit_2014}, often represented as linear parameter-varying (LPV) systems~\cite{wei_gain_2007}. However, physical modelings suffer from limited accuracy and difficulties in parameter identification. In data-driven approaches, LTI modeling has been performed using ARX identification~\cite{iwadare_multi-variable_2009}; however, the applicability of LTI models is restricted due to the wide operating range of engines. For nonlinear system identification, studies such as~\cite{gagliardi_direct_2014} have adopted NARX-based methods, while the literature~\cite{moriyasu_diesel_2019} employed a three-layer NN model. Although these nonlinear modeling approaches achieve high predictive performance, they pose challenges related to computational cost and training complexity. 

\subsection{Contribution and novelty}
This paper presents a bilinear Koopman realization technique from input-output data, incorporating the optimization of the lifting function to enhance its applicability to industrial systems. In the proposed optimization algorithm, RBFs are employed as lifting functions, and their design parameters are optimized using a meta-heuristic approach, specifically PSO. The contributions and novelties of this paper are summarized as follows:

	\emph{Bilinear Koopman realization from input-output data:} For MIMO nonlinear systems, the prediction performance of an LTI Koopman realization is limited; therefore, a bilinear formulation is adopted for Koopman identification. In industrial applications, it is often the case that not all state variables can be measured and only input-output data are available, which is also assumed in this study. By incorporating time-delay coordinates into the measured input-output data and defining the input and output delays as embedded states, an input-output bilinear Koopman realization is constructed. For autonomous systems, it is natural that the embedded states are simply lifted since the states include only the output delays. For controlled systems, the selection of the embedded states lifted is important since the states include input delays in addition to output delays. Prior research (e.g.,~\cite{korda_linear_2018,xiawen_parameters_2021}) employs just embedded states with outputs and inputs delay as arguments of lifting functions. However, the evolution of the inputs is not of interest from the perspective of system identification and feedback control. This is because the inputs are treated as exogenous signals supplied as random or feedback control inputs, rather than as variables evolving autonomously according to a dynamical flow. In this paper, we examine a proper selection of lifting arguments for input-output Koopman realizations. The proposed generalized bilinear Koopman realization solely from input-output data has not been considered in prior works.

	\emph{Optimization of lifting function:} The design of lifting functions significantly influences the prediction performance of Koopman realization and typically requires substantial manual effort~\cite{yahagi_sparse_2025}.  In this study, RBFs are adopted as the lifting functions, and they are optimized using a metaheuristic approach based on multi-step prediction evaluation. This evaluation ensures high predictive accuracy of the Koopman model. The proposed approach has not been explored in prior research and offers several advantages: it does not require prior physical insight, leverages PSO as a well-established global optimization solver, and avoids the inherent challenges of neural network training. Compared to neural network training, PSO offers a robust global optimization framework with easier hyperparameter tuning, clearer objective function definition, and simpler implementation.
     
    \emph{Application to the real-world industrial system:} The proposed method is validated through its application to a diesel engine airpath system, which exhibits MIMO characteristics, strong interaction effects, and nonlinear behavior. The method provides a global bilinear formulation that achieves a balance between predictive accuracy and computational efficiency. This study represents the first application of Koopman-based modeling to a diesel engine airpath system, to the best of the authors' knowledge. Accordingly, the proposed IO-Koopman modeling approach is verified through this real-world industrial application.

\subsection{Article organization}
This paper is organized as follows. Section~\ref{sec:Data-driven modeling via Koopman operator} introduces the basic concept of the Koopman operator for input-dependent nonlinear systems, along with methods for representing LTI and bilinear systems using Koopman matrices with finite-dimensional approximations, and techniques for identifying these matrices. Section~\ref{sec:Proposed methodology} proposes a hyperparameter identification method for RBFs, which are employed as lifting functions, in addition to input-output bilinear Koopman modeling. Sections \ref{sec:Simulation} and \ref{sec:Experimental verification} evaluate the effectiveness of the proposed method through simulation and experimental tests, respectively. The controlled system is the intake and exhaust system of an internal combustion engine, which exhibits strong nonlinearity and MIMO characteristics common in industrial applications. Finally, Section~\ref{sec:Conclusions} summarizes the paper and discusses future research directions.

\subsection{Notation}
The symbol $\mathbb{R}^n$ denotes the set of real numbers in $n$-dimensional space,  and the symbol $\mathbb{R}^{n \times m}$ denotes the space of $n$-row and $m$-column matrices (or two-dimensional arrays) composed of real numbers.
The symbols $\mathcal{M}$ and $\mathcal{N}$ denote smooth manifolds.  
The symbol $\circ$ denotes function composition, the symbol $\odot$ denotes the elementwise product, and the symbol $\otimes$ 
denotes the Kronecker product.
The symbol $\lVert \cdot \rVert_{F}$ denotes the Frobenius norm defined as  
$\lVert A \rVert_{F} = \sqrt{ \sum_{i,j} a_{ij}^{2} }$, where $a_{ij}$ is the element of matrix $A$.
The symbol $I_m$ denotes the identity matrix of size $m \times m$.
The signal sequence is expressed as $(x_i)_{i=0}^{\infty}
$.
The space for all signal sequences is expressed as $\ell(\mathcal{X}) = \{ (x_i)_{i=0}^{\infty} \mid x_i \in \mathcal{X} \}.$
The map $\mathrm{vec}:\ell(\mathcal{X}) \to \mathbb{R}^{\infty}$ convets a signal sequence to vectorized form.
The symbol $0_{m \times n}$ is the zero matrix of size $m \times n$.
The symbol $\dagger$ denotes the Moore-Penrose pseudoinverse of a matrix.
The notation $\psi_{[1:n]}$ is defined as $\psi_{[1:n]} = [\psi_1, \ldots, \psi_n]^\top$,
where $\psi_{[1:n]}$ is a vector-valued function composed of functions $\psi_i$ 
for $i \in \{1,2,\ldots,n\}$.

\section{Data-driven modeling via Koopman operator}
\label{sec:Data-driven modeling via Koopman operator}

\subsection{Koopman operator for non-autonomous system}
\label{subsec:Koopman operator for non-autonomous system}
We describe the Koopman operator theory for non-autonomous systems~ 
\cite{korda_linear_2018,o_williams_kernel-based_2015,proctor_generalizing_2018}.  
Consider a non-autonomous nonlinear system that includes control inputs:
\begin{equation}
    x_{k+1} = f(x_k, u_k)
    \label{eq:nonauto_sys}
\end{equation}
where $f : \mathcal{M} \times \mathcal{N} \to \mathcal{M}$ is the mapping, 
$x_k \in \mathcal{M} \subset \mathbb{R}^m$,
$u_k \in \mathcal{N} \subset \mathbb{R}^m$ denotes the control input at time step $k$,  and $\mathcal{N}$ is a smooth manifold.
For the system, an extended state is defined as
\begin{equation}
    \chi_k = 
    \begin{bmatrix}
        x_k \\
        \nu_k
    \end{bmatrix}
    \label{eq:extended_state}
\end{equation}
where ${\nu_k = \mathrm{vec}\bigl( (u_i)_{i=k}^{\infty} \bigr)}$ denotes the state variable constructed 
from input sequences $(u_i)_{i=k}^{\infty} \in \ell(\mathcal{U})$.
The dynamics of the system with the extended state are given as
\begin{equation}
    \chi_{k+1} = F(\chi_k)
    = 
    \begin{bmatrix}
        f(x_k, \nu_k(0)) \\
        \mathcal{S}  \nu_k
    \end{bmatrix}
    \label{eq:extended_dynamics}
\end{equation}
where $\nu$ is the shift operator defined as ${\mathcal{S} \nu_k = \nu_{k+1}}$  
and ${\nu_k(0) = u_k}$ denotes the first element of $nu_k$ at time $k$.
Therefore, the Koopman operator $\mathcal{K} : \mathcal{F} \to \mathcal{F}$ 
for nonlinear systems with inputs is defined as:
\begin{equation}
    (\mathcal{K}\psi)(\chi_k) = \psi(F(\chi_k)).
    \label{eq:koopman_input}
\end{equation}
The Koopman operator for controlled systems enables a nonlinear system to be represented 
in a (generally infinite-dimensional) linear form by lifting its state space to 
an infinite-dimensional space using the observable.  
The linearity of the Koopman operator helps the use of various linear control theories 
\cite{korda_linear_2018}.

\begin{rem}
The manifolds $\mathcal{M}$ and $\mathcal{N}$ can be considered as 
$\mathcal{M} \subset \mathbb{R}^n$ and $\mathcal{N} \subset \mathbb{R}^m$ 
for most engineering problems \cite{proctor_generalizing_2018,wang_deep_2024}.  
Actually, the states and inputs of the diesel engine airpath system are included in 
$\mathbb{R}^n$ and $\mathbb{R}^m$, as shown in the simulation and experimental test sections.
\end{rem}

\subsection{Finite-dimensional Koopman operator}
\label{subsec:Finite-dimensional Koopman operator}

In Koopman operator theory, the operator acts on an infinite-dimensional space; therefore, 
for numerical computations, it is necessary to introduce a finite-dimensional approximation 
of the Koopman operator.  
In this section, we present the finite-dimensional approximation of the Koopman operator.
The time evolution of a vector-valued lifting function 
$\psi: \mathbb{R}^n \times \mathbb{R}^m \to \mathbb{R}^{p+q}$ 
with $\psi \in \mathcal{F}$ is given as
\begin{equation}
    \psi(x_{k+1}, u_{k+1})
    = K \, \psi(x_k, u_k)
    \label{eq:koopman_finite1}
\end{equation}
where 
$K \in \mathbb{R}^{(p+q) \times (p+q)}$ 
is the finite-dimensional Koopman matrix.
The lifting function $\psi : \mathbb{R}^n \times \mathbb{R}^m \to \mathbb{R}^{p+q}$,  
which corresponds to the dictionary of observables, can be expressed using two lifting 
functions ${\psi_x: \mathbb{R}^n  \to \mathbb{R}^{p}}$ 
and $\psi_u: \mathbb{R}^n \times \mathbb{R}^m \to \mathbb{R}^{q}$ as follows:
\begin{equation}
    \psi(x_k, u_k)
    =
    \begin{bmatrix}
        \psi_x(x_k) \\
        \psi_u(x_k, u_k)
    \end{bmatrix}.
    \label{eq:lifting_split}
\end{equation}
Since the lifting function ${\psi_x}$ depends on the state variables $x$, and $\psi_u$ depends on both the state variables $x$ and the input $u$, the generality is not lost~\cite{son_handling_2020}.
Under this assumption, the time evolution of the lifting function using the 
finite-dimensional Koopman operator (i.e., the Koopman matrix $K$) is given as
\begin{equation}
\begin{aligned}
    \begin{bmatrix}
        \psi_x(x_{k+1}) \\
        \psi_u(x_{k+1}, u_{k+1})
    \end{bmatrix}
    &= 
    K
    \begin{bmatrix}
        \psi_x(x_k) \\
        \psi_u(x_k, u_k)
    \end{bmatrix}
    \\[6pt]
    &=
    \begin{bmatrix}
        K_{11} & K_{12} \\
        K_{21} & K_{22}
    \end{bmatrix}
    \begin{bmatrix}
        \psi_x(x_k) \\
        \psi_u(x_k, u_k)
    \end{bmatrix}
\end{aligned}
\label{eq:koopman_block}
\end{equation}
where 
${K_{11} \in \mathbb{R}^{p \times p}}$, 
${K_{12} \in \mathbb{R}^{p \times q}}$, 
${K_{21} \in \mathbb{R}^{q \times p}}$, 
and 
${K_{22} \in \mathbb{R}^{q \times q}}$.
Since we are interested in the time evolution of the states,  
the above equation can be simplified as follows:
\begin{equation}
    \psi_x(x_{k+1})
    =
    K_{11} \, \psi_x(x_k)
    +
    K_{12} \, \psi_u(x_k, u_k).
    \label{eq:koopman_simplified}
\end{equation}
Previous studies \cite{bruder_advantages_2021,goswami_global_2017} have demonstrated that various system representations
can be achieved by appropriately defining the lifting function $\psi_u(x_k, u_k)$.

\subsubsection{LTI Koopman form}
When representing the system in an LTI Koopman form, setting 
$\psi_u(x_k, u_k) = u_k$ yields the following LTI-Koopman form:
\begin{equation}
    \psi_x(x_{k+1})
    =
    K_{11}\,\psi_x(x_k)
    +
    K_{12}\,u_k.
    \label{eq:LTI_Koopman}
\end{equation}
Here, let $A = K_{11}$,  $B = K_{12}$,  
and $z_k = \psi_x(\cdot)$.  
Then, the LTI state-space equation for lifted states $z_k$ is given as
\begin{equation}
    z_{k+1} = A z_k + B u_k.
    \label{eq:LTI_state}
\end{equation}

\subsubsection{Bilinear Koopman form}
\label{subsec:Bilinear Koopman form}
According to the literature \cite{bruder_advantages_2021,abtahi_deep_2025},  
a bilinear Koopman form can be expressed by defining the lifting function  
${\psi_u : \mathbb{R}^n \times \mathbb{R}^m \to \mathbb{R}^{m+mp}}$ as follows:
\begin{equation}
    \psi_u(x_k, u_k)
    =
    \begin{bmatrix}
        u_k \\
        \psi_x(x_k) \otimes u_k
    \end{bmatrix}
    \label{eq:bilinear_lifting}
\end{equation}
where $\psi_x(x_k)\otimes u_k$ is calculated as
\begin{equation}
\begin{split}
\psi_x(x_k) \otimes u_k
=
\bigl[
    &u_k(1)\psi_x(x_k),\;
    u_k(2)\psi_x(x_k),\;
    \ldots\\
    &\ldots
    ,u_k(m)\psi_x(x_k)
\bigr]^\top
\in \mathbb{R}^{mp}.
\end{split}
\end{equation}
Then, the time evolution of the lifting functions is expressed as
\begin{equation}
\begin{bmatrix}
\psi_x(x_{k+1}) \\
u_{k+1} \\
\psi_x(x_{k+1}) \otimes u_{k+1}
\end{bmatrix}
=
\begin{bmatrix}
A & B_0 & B \\
\ast & \ast & \ast \\
\ast & \ast & \ast
\end{bmatrix}
\begin{bmatrix}
\psi_x(x_k) \\
u_k \\
\psi_x(x_k) \otimes u_k
\end{bmatrix}
\label{eq:bilinear_matrix_form}
\end{equation}
where $A \in \mathbb{R}^{p \times p}$, $B_0 \in \mathbb{R}^{p \times m}$, 
$B = \bigl[B_1,\  \cdots ,\ B_m\bigr] \in \mathbb{R}^{p \times mp}$, and 
$B_i \in \mathbb{R}^{p \times p}$ $(i=1,\ldots,m)$. 
Since we are interested in the time evolution of the states and not in the time evolution of the inputs, 
the elements of the Koopman matrix related to the input dynamics can be ignored. 
Therefore, the above equation (\ref{eq:bilinear_matrix_form}) is reduced to
\begin{equation}
\psi_x(x_{k+1}) = A \, \psi_x(x_k) + B_0 \, u_k + B \bigl(\psi_x(x_k) \otimes u_k\bigr).
\end{equation}
Let $z_k = \psi_x(\cdot)$ denote the lifted state at time $k$.  
Then, the bilinear state equation can be written in the following form:

\begin{equation}
\begin{aligned}
    z_{k+1}
    &=
    Az_k
    +
    B_0u_k
    +
    B(z_k \otimes u_k)
    .
\end{aligned}
    \label{eq:bilinear_equilibrium}
\end{equation}

\subsection{Learning of Koopman matrix}
\label{subsec:Learning_of_Koopman_matrix}
For a nonlinear system with inputs, we assume to measure time-series data of the measurable 
state variables $x_k$, and the inputs $u_k$ are collected, i.e.,
${\mathcal{D} = \{ (x_k, u_k) \mid k = 1, \ldots, N_d \}}$.
From the dataset, the snapshot data are set as:
\(Z^{+}\) and \(Z_{w}\) are given as
\begin{equation}
  \begin{aligned}   
  Z^{+} &= \bigl[\, z_2,\; z_3,\; \ldots,\; z_{N_d} \,\bigr]\\
  Z_{w} &=
  \begin{bmatrix}
    z_1      & z_2      & \cdots & z_{N_d-1} \\
    w_1      & w_2      & \cdots & w_{N_d-1} \\
    z_1\!\otimes\! w_1 & z_2\!\otimes\! w_2 & \cdots & z_{N_d-1}\!\otimes\! w_{N_d-1}
  \end{bmatrix}.
  \end{aligned}
\end{equation}
Then, the state-space matrices $A$, $B_0$, $B$  
for the bilinear Koopman realization are obtained from the following optimization problem:
\begin{equation}
  \min_{H}\;\bigl\| Z^{+} - H\, Z_{w} \bigr\|_{F}
  \label{eq:opt_problem_for_bilinear_realization}
\end{equation}
where \(H=[\,A,\;B_0,\;B\,] \in \mathbb{R}^{p\times p(m+l+1)}\).
This least-squares problem can be
solved 
using the Moore-Penrose pseudoinverse as follows:
\begin{equation}
    H = Z^+ [Z_w, U]^{\dagger}.
    \label{eq:LS_solution}
\end{equation}
In the case of a linear system, the elements of the bilinear term are removed from the snapshot data.

\section{Proposed methodology}
\label{sec:Proposed methodology}

\subsection{Time-delay coordinates}

We explain input-output Koopman modeling \cite{korda_linear_2018}, which addresses the case where only 
input-output data can be measured while the system states are unknown or unobservable.  
Since the diesel engine airpath system which has exogenous inputs is considered in this paper, 
this section handles the nonlinear system with control and exogenous inputs, given as
\begin{equation}
    \begin{aligned}
        x_{k+1} &= f(x_k, u_k, d_k) \\
        y_k &= h(x_k, u_k, d_k)
    \end{aligned}
    \label{eq:nonlinear_io}
\end{equation}
where \( d_k \in \mathbb{R}^{l} \) is the measurable exogenous input,  
\( y_k \in \mathbb{R}^{n_h} \) is the measurable output, and  
\( f : \mathbb{R}^n \times \mathbb{R}^m \times \mathbb{R}^{l} \to \mathbb{R}^n \),  
\( h : \mathbb{R}^n \times \mathbb{R}^m \times \mathbb{R}^{l} \to \mathbb{R}^{n_h} \).
For this system with inputs, we assume to obtain the dataset:
$\mathcal{D} = \{ y_k, w_k \mid k = 1, \ldots, N_d \}$, 
where \( w_k \) is defined as  
\( w_k = [ u_k \;\; d_k ]^\top \).
The time-delay coordinates are introduced to address the constraint that only input-output 
data can be measured.  
The \emph{embedded} state at time \( k \), considering delay-step \( n_d \), is defined as:
\begin{equation}
    \zeta_k = 
    \begin{bmatrix}
        (\zeta_k^y)^\top & (\zeta_k^w)^\top
    \end{bmatrix}^\top
    \in \mathbb{R}^{n_\zeta}
    \label{eq:embedded_state}
\end{equation}
with
\begin{gather}
    \zeta_k^y =
    \begin{bmatrix}
        y_k^\top & y_{k-1}^\top & \cdots & y_{k-n_d}^\top
    \end{bmatrix}^\top
    \in \mathbb{R}^{(n_d+1)n_h}
        \label{eq:zeta_y}
    \\
    \zeta_k^w =
    \begin{bmatrix}
        w_{k-1}^\top & w_{k-2}^\top & \cdots & w_{k-n_d}^\top
    \end{bmatrix}^\top
    \in \mathbb{R}^{n_d(m+l)}
    \label{eq:zeta_w}
\end{gather}
where 
\(
    n_\zeta = n_d(n_h+m+l)n_h
\).

\subsection{IO generalized bilinear Koopman form}

As well as Section~\ref{subsec:Finite-dimensional Koopman operator}, the Koopman model is given as
\begin{equation}
    \psi_x(\zeta_{k+1}) = 
    K_{11} \psi_x(\zeta_k)
    +
    K_{12} \psi_w(\zeta_k)
    \label{eq:io_koopman}
\end{equation}
where the lifting function 
\( \psi_w : \mathbb{R}^{n_\zeta} \times \mathbb{R}^{n_d(m+l)} \to \mathbb{R}^{(m+l)+(m+l)q} \) 
is defined as:
\begin{equation}
    \psi_w(\zeta_k, w_k) =
    \begin{bmatrix}
        w_k \\
        \phi_w(\zeta_k) \otimes w_k
    \end{bmatrix}
    \label{eq:psi_w}
\end{equation}
where \( \phi_w : \mathbb{R}^{n_\zeta} \to \mathbb{R}^{n_w} \) is the function.
Then, analogous to (\ref{eq:bilinear_matrix_form}), the standard form is formulated as
\begin{equation}
\begin{bmatrix}
    \psi_x(\zeta_{k+1})\\
    w_{k+1} \\
    \phi_w(\zeta_{k+1}) \otimes w_{k+1}
\end{bmatrix}
=
\begin{bmatrix}
    A & B_0 & B \\
    \ast & \ast & \ast \\
    \ast & \ast & \ast
\end{bmatrix}
\begin{bmatrix}
    \psi_x(\zeta_k) \\
    w_k \\
    \phi_w(\zeta_k) \otimes w_k
\end{bmatrix}
\end{equation}
Since the time evolution of interest corresponds to the embedded states,  
the above equation is simplified as:
\begin{equation}
    \psi_x(\zeta_{k+1})
    =
    A \psi_x(\zeta_k)
    +
    B_0 w_k
    +
    B \left( \phi_w(\zeta_k) \otimes w_k \right).
    \label{eq:io_general_bilinear}
\end{equation}
By defining 
% \(
%     z_k = \psi_x(\zeta_k)
% \)
% and 
\(
    z_{w,k} = \phi_w(\cdot),
\)
the generalized bilinear forms can be rewritten by
\begin{equation}
    z_{k+1}
    =
    A z_k
    +
    B_0 w_k
    +
    B (z_{w,k}\otimes w_k).
    \label{eq:bilinear_rewritten}
\end{equation}

\begin{rem}
The IO generalized bilinear Koopman form \eqref{eq:io_general_bilinear}  
can be viewed as a bilinear-like LPV Koopman realization  
by considering  
\(
    \phi_w(\zeta_k)
\)
as a scheduling function.  
\end{rem}

\subsection{Argument selection of lifting function for IO-Koopman realization}
The lifting function in previous studies on the standard IO-Koopman form is defined as

\begin{equation}
% z_k = 
\psi_x(\zeta_k)
=
\begin{bmatrix}
\zeta_k \\
\psi_{[1:n_l]}(\zeta_k) \\
\end{bmatrix}
\in \mathbb{R}^{n_\zeta + n_l}
\label{eq:standard_IO_lifting_function}
\end{equation}
where ${\psi_{[1:n_l]} = [\psi_1, \ldots, \psi_{n_l}]^\top}$, and $\psi_i$ is a scalar function for ${\ i \in \{1, \ldots, n_l\}}$.
In previous IO-Koopman realization \cite{korda_linear_2018}, the lifting function is typically chosen as $\psi_i(\zeta_k)$.
Under this lifting function setting, if $\psi(\zeta_k) = [y_k,\ w_{k-1},\ w_{k-1} {\odot} w_{k-1}]^\top$ is given,
the time evolution of $\psi(\zeta_k)$ includes inputs.
However, we are not interested in predicting observables that include inputs,
since inputs are externally provided---either randomly or through a control law.
Thus, employing $\psi(\zeta_k)$ may result in the failure of capturing the essential
system dynamic characteristics due to spurious dynamics of exogenous input,
which deteriorate the prediction performance or require a significantly larger
amount of training data.
Although infinite input sequences are assumed in previous research,
it is not feasible to cover all possible input sequences due to the finite amount of
data in practical situations.
Therefore, unlike prior studies~\cite{korda_linear_2018,xiawen_parameters_2021}, we propose explicitly excluding $w_k$ from nonlinear lifting to prevent the identification of spurious dynamics induced by exogenous inputs. Specifically, we propose the following observable selection:
\begin{equation}
% z_k = 
\psi_x(\zeta_k) 
=
\begin{bmatrix}
\zeta_k \\
\psi_{[1:n_l]}(\zeta_k^y) \\
\end{bmatrix}
\in \mathbb{R}^{n_\zeta + n_l}.
\label{eq:}
\end{equation}
%
% \vspace*{\fill} 
\begin{figure*}[!b]%bht
% \begin{strip}
% \centering
\normalsize 
\hrulefill
\begin{equation}
\begin{aligned}
\begin{bmatrix}
y_{k+1}\\
y_{k}\\
\vdots\\
y_{k-n_d+1}\\
w_{k}\\
w_{k-1}\\
\vdots\\
w_{k-n_d+1}\\
\psi_{[1:n_l]}(\zeta^{y}_{k+1})
\end{bmatrix}
=
\begin{bmatrix}
A_{yy1} & A_{yy2} & \cdots & A_{yy n_d} &
A_{yw1} & A_{yw2} & \cdots & A_{y w n_d} & A_{y\psi}
\\[4pt]
I & 0 & 0 & 0 & 0 & 0 & 0 & 0 & 0
\\
0 & \ddots & 0 & \ddots & 0 & 0 & 0 & 0 & 0
\\
0 & 0 & I & 0 & 0 & 0 & 0 & 0 & 0
\\
0 & 0 & 0 & 0 & 0 & 0 & 0 & 0 & 0
\\
0 & 0 & 0 & 0 & I & 0 & 0 & 0 & 0
\\
\vdots & \ddots & \ddots & \vdots & 0 & \ddots & 0 & 0 & 0
\\
0 & 0 & 0 & 0 & 0 & 0 & I & 0 & 0
\\[4pt]
A_{\psi y1} & A_{\psi y2} & \cdots & A_{\psi y n_d} &
A_{\psi w1} & A_{\psi w2} & \cdots & A_{\psi w n_d} & A_{\psi\psi}
\end{bmatrix}
\begin{bmatrix}
y_{k}\\
y_{k-1}\\
\vdots\\
y_{k-n_d}\\
w_{k-1}\\
w_{k-2}\\
\vdots\\
w_{k-n_d}\\
\psi_{[1:n_l]}(\zeta^{y}_{k})
\end{bmatrix}
& \,+\,
\begin{bmatrix}
0\\
0\\
\vdots\\
0\\
I\\
0\\
\vdots\\
0\\
0\\
\end{bmatrix} w_k
\\
& \kern-2em  +\, B\bigl(\phi_w(\zeta_k)\otimes w_k\bigr)
\end{aligned}
\label{eq:transision_time-delay}
\end{equation}
% \end{strip}
\end{figure*}%
Then, the IO bilinear Koopman model is expressed as (\ref{eq:transision_time-delay}) shown at the
bottom of the page.
This formulation indicates that the embedded state with input delay is generated by a given input.
In other words, the input after transitions, $w_k$, is included in the time-evolved state 
\(z_{k+1} (= \psi_x(\zeta^y_{k+1}))\)
and the input delay is generated through this input transition.
Therefore, the matrices $A$ and $B$ associated with the input-output delay transition are trivially determined by the shift relationship.
If there is an observable associated with the input (e.g., $y_k {\odot} w_{k-1}{\odot} w_{k-1}$),  
the observable is not generated from the given input $w_k$.
Thus, the matrix $A_{\psi \star \star}$ in the Koopman matrix $A$ associated with the time evolution of the lifted state
is not uniquely determined.
Therefore, in this study, the argument of $\psi_i$ is $\zeta_k^y$ instead of $\zeta_k$.
On the other hand, because $\phi_w$ does not contribute to the time evolution of $\psi_x(\zeta_k)$,
we employ a lifting function related to $\zeta_k$ to increase the flexibility of our representation:
\begin{equation}
\phi_w(\zeta_k)
=
\begin{bmatrix}
\phi_{[1:n_w]}(\zeta_k) \\
\end{bmatrix}
\in \mathbb{R}^{n_w}
\end{equation}
where \( {\phi_i:\mathbb{R}^{n_\zeta} \to \mathbb{R}}\) for \({i\in \{1,\cdots,n_w\} }\).
The output equation is also given as \({y_k = C z_k}\)
with
$C =
\begin{bmatrix}
I_{n_h} & 0_{n_h \times (n_l - n_h)}
\end{bmatrix}
$.
We summarilize the above discussion as follows:
%%%%%%%%%%%%%%%%%%%%%%%%%%%%%%%%%
\begin{theorem}[]
\label{theorem1}
Consider the IO generalized bilinear Koopman realization (\ref{eq:io_general_bilinear})
with the state-transition structure 
in (\ref{eq:transision_time-delay}), where $z_k(=\psi_x(\cdot))$ represents “state’’ variables updated by a time-invariant
linear map $A$, and the new input $w_k$ is injected via $B_0$ and $B$. 
Suppose that $\psi_x$ contains an input-dependent nonlinear component $\psi_i(\zeta_k)$, such as $w_{k-1} {\odot} w_{k-1}$ or $y_k\odot w_{k-1}$. Then the next-time quantity $g_{k+1}(:= \psi_i(\zeta_{k+1}))$ corresponding to $\psi_i$ generally depends on the new input $w_k$, and it is structurally impossible to generate $g_{k+1}$ by the block $A_{\psi\star\star}$, 
and including such nonlinearities in $\psi_x$ is causally inconsistent with (\ref{eq:transision_time-delay}).
\end{theorem}

\begin{IEEEproof}[Proof]
Assume \(g_{k+1} = a^\top z_k\), 
where \(a^\top\) is the row of \(A_{\psi \star\star}\) corresponding to the observable \(\psi_i\). 
Fix $\zeta_k$ (thus $z_k$) and choose $w_k^{(1)}\neq w_k^{(2)}$. 
Since $g_{k+1}$ depends on $w_k$ (e.g., 
$y_{k+1}\odot w_k$), 
one obtains $g_{k+1}^{(1)}\neq g_{k+1}^{(2)}$, whereas \(a^\top z_k\)
is identical for both choices ($\because$ it contains no $w_k$).
This contradiction proves that $g_{k+1}$ cannot be produced  by $A_{\psi\star\star}$.
\end{IEEEproof}

\begin{rem}
Theorem~\ref{theorem1} indicates that using \(\psi_i(\zeta_k)\) may prevent 
% the essential dynamic characteristics of the system 
the essential system dynamics
from being captured due to spurious dynamics induced by the exogenous input.
\end{rem}

\begin{corollary}
To preserve causal consistency with the state-transition structure in (\ref{eq:transision_time-delay}), input-dependent nonlinearities must be allocated to $\phi_w(\zeta_k)$, and the arguments of $\psi_i$ should be restricted to $\zeta^y_k$. 
\end{corollary}

\begin{IEEEproof}
According to Theorem \ref{theorem1}, placing input nonlinearities inside $\psi_x$ forces their next-time image to depend on $w_k$, which cannot be realized by the state-transient matrix $A$. 
Thus, to preserve the interpretation that $\psi_x$ represents state variables updated solely by $A$, input-dependent nonlinearities should be handled through $\phi_w(\zeta_k)$ so that their $w_k$-dependence is  routed via $B_0$ and $B$, while $\psi_x$ is restricted to $\zeta^y_{k}$. 
\end{IEEEproof}

\subsection{IO generalized Koopman realization}
\label{subsec:io_generalized_koopman_realization}
This section describes IO generalized bilinear Koopman realization. 
The snapshot data for the extended states and inputs are set as
\begin{gather}
      X = \bigl[\, \zeta_1,\; \zeta_2,\; \ldots,\; \zeta_{N_d-n_d} \,\bigr]\\
      W = \bigl[\, w_1,\; w_2,\; \ldots,\; w_{N_d-n_d} \,\bigr]
  % \quad
  \label{eq:41}
\end{gather}
with
\begin{gather}
  \zeta_i
  = 
  \begin{bmatrix} (\zeta_i^y)^\top  & (\zeta_i^w)^\top \end{bmatrix}^\top \in \mathbb{R}^{n_\zeta}\\
  \zeta_i^{\,y}
  =
  \begin{bmatrix}
    y_{i,n_d}^\top &
    y_{i,n_d-1}^\top &
    \cdots &
    y_{i,0}^\top
  \end{bmatrix}^{\!\top}
  \in \mathbb{R}^{(n_d+1)n_l}\\
  \zeta_i^{\,w}
  =
  \begin{bmatrix}
    w_{i,n_d-1}^\top &
    w_{i,n_d-2}^\top &
    \cdots &
    w_{i,0}^\top
  \end{bmatrix}^{\!\top}
  \in \mathbb{R}^{n_d(m+l)}.
\end{gather}
Then, the snapshot data 
\(Z^{+}\) and \(Z_{w}\) 
for IO generalized koopman realization 
are expressed as
\begin{equation}
  \begin{aligned}   
  Z^{+} &= \bigl[\, z_2,\; z_3,\; \ldots,\; z_{N_d} \,\bigr]\\
  Z_{w} &=
  \begin{bmatrix}
    z_1      & z_2      & \cdots & z_{N_d-1} \\
    w_1      & w_2      & \cdots & w_{N_d-1} \\
    z_{w,1}\!\otimes\! w_1 & z_{w,2}\!\otimes\! w_2 & \cdots & z_{w,N_d-1}\!\otimes\! w_{N_d-1}
  \end{bmatrix}
  \end{aligned}
\end{equation}
and, the IO generalized bilinear Koopman model is obtained from the optimization probelm  (\ref{eq:opt_problem_for_bilinear_realization}).

\subsection{Optimization of lifting functions}
\label{subsec:opt_lifting}
This section presents the optimization method of lifting functions. Herein, we optimize the design parameters of RBFs. The lifting function is generally selected, such as mechanics-inspired selection, empirical selection (e.g., monomials, polynomials), and RBF selection with randomly selected design parameters. Such selections may not provide good performance since the lifting function is not optimized. The finding method of the proper lifting function is mainly a NN-based approach. This approach has the same difficulties as general NN learning, including hyperparameter setting, overfitting, model complexity, computational cost, vanishing and exploding gradient problem. In this paper, we adopt a metaheuristic optimization (e.g., PSO) approach, as an alternative approach to the NN-based method. 

The optimization algorithm is shown in Algorithm~\ref{algorithm: Optimization of lifting functions}.
In the proposed optimization method, RBFs are adopted as the lifting functions.
For example, polyharmonic RBFs are given as
\begin{gather}
    \psi(x_1;\theta_1)
    = \lVert x_1 - \theta_1 \rVert_2 \log \lVert x_1 - \theta_1 \rVert_2
    \\
    \phi(x_2;\theta_2)
    = \lVert x_2 - \theta_2 \rVert_2 \log \lVert x_2 - \theta_2 \rVert_2
\end{gather}
where \(x_i\) and \(\theta_i\) (\(i=1,2\)) are the input and center vectors of the RBFs.
The center is considered the design parameter.
Then, the design parameters of the RBFs are optimized by global optimization such as PSO.

As shown in Algorithm~\ref{algorithm: Optimization of lifting functions}, the proposed optimization process of the centers of RBFs
encompasses the realization of the Koopman matrix.
The optimization problem is formulated as
\begin{equation}
    \min_{\theta} J(\theta), \;
    J(\theta) = \sum_{k=1}^{N_d} (y_k - \hat{y}_k(\theta))^2
\end{equation}
where \(\hat{y}_k\) denotes \emph{multi-step (not one-step)} prediction given by
\begin{equation}
\left\{
\begin{aligned}
    \hat{z}_{k+1} &= A \hat{z}_k + B_0 w_k + B(\hat{z}_{w,k} \otimes w_k) \\
    \hat{y}_k &= C \hat{z}_k
\end{aligned}
\right.
\end{equation}
with the initial lifted states 
$ z_0 = \psi(\zeta_0)$ and $z_{w,0} = \phi(\zeta_0)$.
Matrices \(A, B_0, B, C\) can be obtained by Koopman realization
for candidate centers, as shown in Section~\ref{subsec:io_generalized_koopman_realization}.

\begin{algorithm}[t]
\caption{Optimization of lifting functions}
\label{algorithm: Optimization of lifting functions}

\begin{algorithmic}[1]

\STATE \textbf{Inputs:}
\STATE \quad Max iteration $N_{\max}$
\STATE \quad RBFs $\psi(x_1;\theta_1),\;\phi(x_2;\theta_2)$
\STATE \quad Dataset $\mathcal{D}=\{y_k,u_k,d_k\mid k=1,\ldots,N_d\}$

\STATE \textbf{Outputs:}
\STATE \quad Design parameter vector $\theta=[\theta_1,\theta_2]$
\STATE \quad Koopman matrices $A, B_0, B, C$

\STATE Set initial $\theta$, Maximum iterations $N_{\max}$
\STATE $iteration \gets 1$

\WHILE{(iteration $\le N_{\max}$)}
  \STATE Compute Koopman matrices $A, B_0, B, C$ for
  $\psi(x_1;\theta_1), \phi(x_2;\theta_2)$
  via Koopman realization
  \STATE Update $\theta$ by minimizing $J(\theta)$
  \IF{$J(\theta)<\varepsilon$}
      \STATE \textbf{break}
  \ENDIF
  \STATE $iteration \gets iteration + 1$
\ENDWHILE

\STATE Return $\theta^\ast = \arg\min J(\theta)$ and matrices $A,B_0,B,C$

\end{algorithmic}
\end{algorithm}

\subsection{Model validation}

In model validation, the modeling accuracy is evaluated using the coefficient
of determination, $R^{2}$, of $y$, defined as
\begin{equation}
    R^{2}
    =
    1 -
    \frac{
        \sum_{k=1}^{N_d} \left( y_k - \hat{y}(k) \right)^2
    }{
        \sum_{k=1}^{N_d} \left( y(k) - \bar{y} \right)^2
    }
\end{equation}
where $\bar{y}$ is the mean value of the $N_d$ data~\cite{haber_structure_1990}.
An $R^2$ score equal to 1.0 indicates that the identified model best fits the target system. From an engineering perspective, an R-squared score is acceptable when it reaches a value from 0.9 to 1.0~\cite{schaible_fuzzy_1997}. A negative value means that the inferred model has a very low ability to represent an equivalent dynamical system. In nonlinear systems, multi-step prediction may not be achievable even when accurate one-step prediction is achieved~\cite{xiao_deep_2023}
. Therefore, we evaluate both one-step and multi-step predictions. The one-step prediction is defined as
\begin{equation}
\left\{
\begin{aligned}
    \hat{x}(t+1) &= f\bigl(x(t),\, u(t),\, d(t)\bigr)\\
    \hat{y}(t) &= h\bigl(\hat{x}(t)\bigr)
\end{aligned}
\right.
\end{equation}
where $\hat{y}$ is the predicted value of $y$. 
The one-step ahead is predicted from the given input and output data.
The multi-step prediction is defined as
\begin{equation}
\left\{
\begin{aligned}
    \hat{x}(t+1) &= f\bigl(\hat{x}(t),\, u(t),\, d(t)\bigr) \\
    \hat{y}(t) &= h\bigl(\hat{x}(t)\bigr)
\end{aligned}
\right.
\end{equation}
with $\hat{x}(0)=x(0)$. 
The multi-step ahead is predicted from the given input and the predicted past output.

\section{Simulation}
\label{sec:Simulation}

\subsection{Target system}
\label{subsec:Target system}

The intake and exhaust system of a four-cylinder diesel engine~\cite{yahagi_sparse_2025} is shown in Fig.~\ref{fig:Diesel engine airpath system.png}. 
This system includes an exhaust gas recirculation (EGR) system and a variable geometry turbocharger (VGT). 
The EGR system regulates the oxygen concentration entering the cylinders by mixing oxygen-lean exhaust gas with fresh air. 
This reduces the oxygen concentration, lowers the peak combustion temperature, and suppresses the formation of harmful nitrogen oxides, 
which are produced in large quantities at high temperatures.
The VGT system adjusts the pressure inside the intake manifold. By narrowing the spacing of the movable vanes, 
the flow path is restricted and the flow rate increases, achieving a supercharging effect even at low engine speeds. 
At high engine speeds, the vane spacing is widened to facilitate exhaust gas flow.
Fresh air is drawn in from the atmosphere and compressed by the compressor. After cooling in the intercooler, 
it passes through the intake throttle and enters the intake manifold. 
Exhaust gas also flows into the intake manifold after passing through the EGR valve. 
The oxygen concentration in the intake manifold is controlled by operating the EGR valve. 
The mixture then flows into the cylinders, and after combustion, the exhaust gas is discharged into the exhaust manifold. 
The gas leaving the exhaust manifold splits into two paths: one recirculated through the EGR cooler, 
and the other passing through the turbine and exiting as exhaust gas. 
The burned gas drives the turbine, and the turbine speed is controlled by adjusting the vane angle.
In this study, modeling is performed using a data-driven approach; therefore, further explanation of physical modeling is omitted. 
For details, refer to~\cite{kaneko_model_2019}.

\begin{figure}[tb]%thb
\centering
\includegraphics[width=4.5cm]{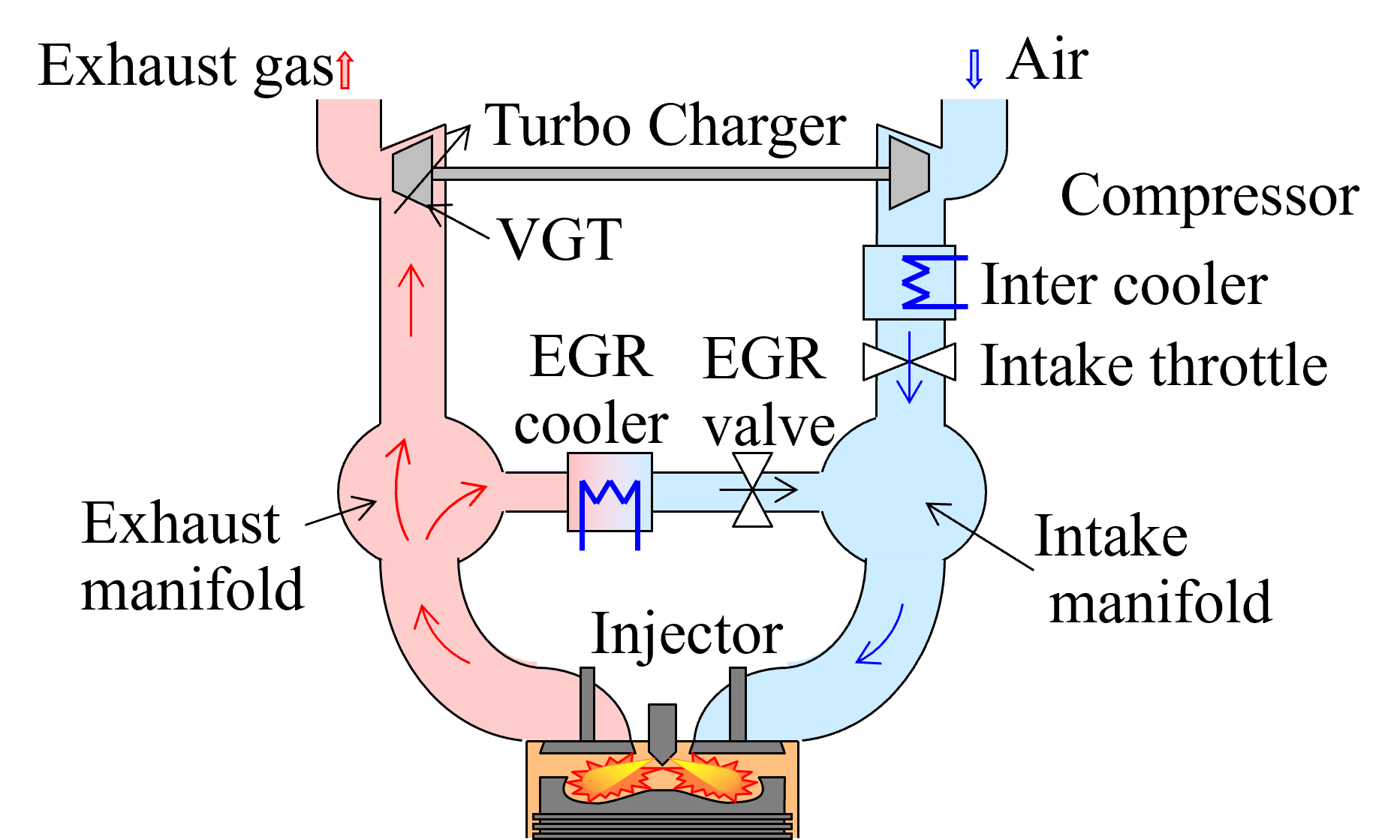}
\vspace{-0.5em} % 上に近づける
\caption{Diesel engine airpath system.
\label{fig:Diesel engine airpath system.png}}
% \end{center}
\end{figure}

\subsection{Simulation setting}
\label{subsec:Simulation setting}

Training and test data were generated using a design of experiments (DoE). Refer to~\cite{yahagi_sparse_2025,yonezawa_sparse_2026}. 
The generated input signals were applied to the plant inputs (i.e., VGT valve position, EGR valve position, fuel injection amount, and engine speed), and the outputs (i.e., intake manifold pressure and EGR rate) were measured. The sampling period was set to 0.1 seconds, resulting in 25,000 data points. The time-series data were normalized as a preprocessing step. To optimize the RBF hyperparameters, PSO, a reliable global optimization solver, was employed. A delay step of 2 was used.
PSO was implemented using the MATLAB's \texttt{particleswarm} function with a maximum of 30 iterations and a swarm size of 300. Parallel computation was also utilized. 
Table~\ref{tab:simulation_cases} shows the koopman forms considered. 
The number of optimization parameters was set to 10 for all cases. 
In the table,
    H-DMD represents Hankel-DMD \cite{zawacki_dynamic_2023}: $z_{k+1} = Az_k + B_w w_k$,
     LK represents the LTI Koopman realization \cite{korda_linear_2018,williams_extending_2016}: $z_{k+1} = A z_k + B_w w_k$,
   BLK represents the standard bilinear Koopman realization\cite{bruder_advantages_2021,yahagi_sparse_2025}: 
    $z_{k+1} = A z_k + B_w w_k + B_{z} (z_k \otimes w_k)$,
    GBK represents the generalized bilinear Koopman realization (the proposed method): 
    $z_{tk+1} = A z_t + B_w w_t + B_z (z_{w,k} \otimes w_k)$.

%%%%% Table 1
\begin{table}[tb]
\centering
\caption{Simulation cases of Koopman realization. 
}
\label{tab:simulation_cases}
\renewcommand{\arraystretch}{1.15} 
\vspace{-1.0em} % 上に近づける
\begin{tabular}{c|c|c|c}
\hline
\textbf{Case} &
\centering \makecell{\textbf{Lifting }\\\textbf{function $\psi_x$}}&
\centering \makecell{\textbf{Lifting }\\\textbf{function $\phi_w$}} &
\textbf{Opt.}
\\ \hline

H-DMD &
\centering $\zeta_k$ &
\centering -- &
-- 
\\ \hline

LK &
\centering $\begin{bmatrix}
\zeta_k \\ \psi_{[1:10]}(\zeta_k)
\end{bmatrix}$ &
\centering -- &
with
\\ \hline

BLK(poly) &
10th-order polynomial of $\zeta_k$ &
Similar to $\psi_x$ &
-- 
\\ \hline

BLK &
\centering $\begin{bmatrix}
\zeta_k \\ \psi_{[1:10]}(\zeta_k)
\end{bmatrix}$ &
Similar to $\psi_x$ &
with
\\ \hline

GBLK &
\centering $\begin{bmatrix} 
\zeta_k \\ \psi_{[1:5]}(\zeta^y_k)
\end{bmatrix}$ &
\centering $\begin{bmatrix}
\zeta_k \\ \phi_{[1:5]}(\zeta_k)
\end{bmatrix}$ &
with
\\ \hline
\end{tabular}
\end{table}

\subsection{Results and discusstions}
\label{subsec:Simulation results}

Fig.~\ref{fig:SimDE_lrn_data} illustrates the training and test data generated by applying input signals and measuring the corresponding outputs. 
The orange and blue lines represent the training and test datasets, respectively. 
In the figure, $y_1$, $y_2$, $u_1$, $u_2$, $d_1$, and $d_2$ denote the intake manifold pressure~[kPa], 
EGR rate~[\%], VGT vane position~[\%], EGR valve position~[\%], fuel injection amount~[mm\textsuperscript{3}/st], and engine speed~[rpm], respectively. 
Fig.~\ref{fig:SimDE_predict_valid_data} presents the time-series data of prediction performance for test data. 
The dotted lines represent the predicted values, and the solid lines represent actual values. 
Table~\ref{tab:sim_R2_longterm} reports the R-squared scores for one-step and long-term predictions for the above cases. We compared our method with conventional approaches, including Hankel DMD, the LTI Koopman model, the bilinear Koopman model with polynomial basis functions, and the bilinear Koopman model with randomly assigned RBF centers. Among these, Hankel DMD without lifting exhibited the poorest predictive performance, and the LTI Koopman model without RBF optimization also performed poorly. Although RBF optimization significantly improved the predictive accuracy of the LTI Koopman model, the improvement was still insufficient, highlighting the need for identification in the bilinear model. Next, we examined the case where the lifting function for the bilinear form was polynomial and found that its predictive performance was substantially poor. For the standard bilinear form, the $R^2$ value exceeded 0.8 after RBF optimization. However, increasing the number of RBFs led to a marked deterioration in predictive performance on the test data, indicating poor generalization. This is likely due to the inclusion of inputs in the observables, as discussed in Section~\ref{subsec:io_generalized_koopman_realization}. Finally, the proposed generalized bilinear form B with RBF optimization achieved the highest prediction accuracy in this study, with an $R^2$ value of over 0.9, meeting the requirements for engineering applications~\cite{schaible_fuzzy_1997}.

%%%% Table2
\begin{table}[tb]
\centering
\caption{R-squared of one-step and long-term predictions in the simulation.}
\label{tab:sim_R2_longterm}

\vspace{-1.0em} % 上に近づける
\begin{tabular}{c|c|c|c|c|c|c}
\hline
\multirow{3}{*}{\textbf{Case}} &
\multicolumn{2}{c|}{\makecell{\textbf{One-step}\\\textbf{prediction}}} &
\multicolumn{4}{c}{\textbf{Long-term prediction}} \\ \cline{2-7}
& \multicolumn{2}{c|}{learning data} &
\multicolumn{2}{c|}{learning data} &
\multicolumn{2}{c}{test data} \\ \cline{2-7}
& $y_1$ & $y_2$ & $y_1$ & $y_2$ & $y_1$ & $y_2$ \\ \hline
H-DMD & 0.994 & 0.991 & $-14.5$ & $-7.85$ & $-15.0$ & $-7.89$ \\ \hline
LK  & 0.992 & 0.989 & 0.793  & 0.910  & 0.778  & 0.908  \\ \hline
BLK(poly)  & 0.996 & 0.993 & $-14.3$ & $-7.77$ & $-14.8$ & $-7.78$ \\ \hline
BLK  & 0.996 & 0.993 & 0.874  & 0.988  & $-0.295$ & 0.347  \\ \hline
GBLK & \textbf{0.995} & \textbf{0.992} & \textbf{0.931}  & \textbf{0.951}  & \textbf{0.922}  & \textbf{0.949}  \\ \hline
\end{tabular}
\end{table}

\begin{figure}[tb]%thb
\centering
\includegraphics[width=8.0cm]{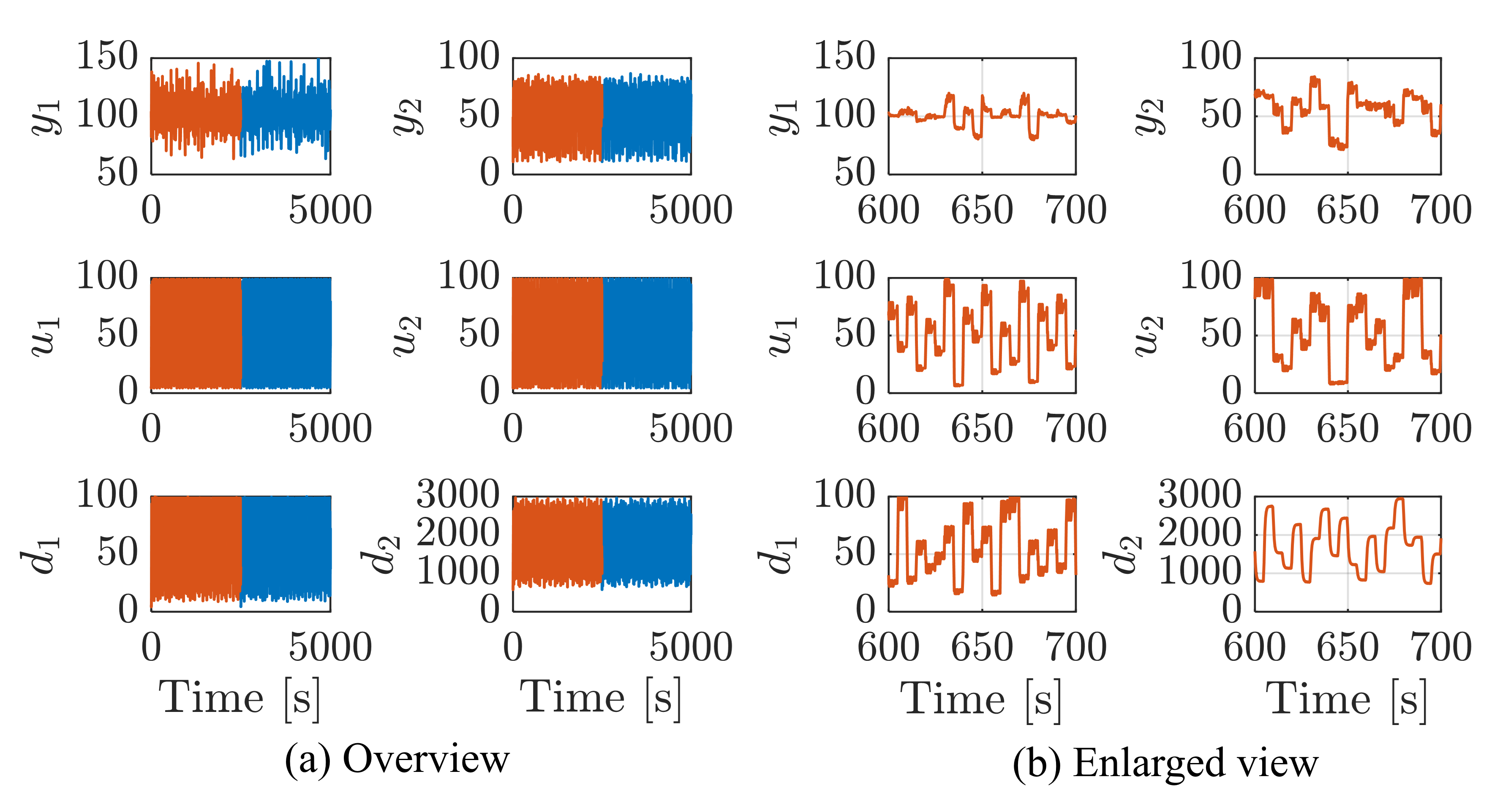}
\vspace{-1em} 
\caption{The learning and test data in simulation. 
\label{fig:SimDE_lrn_data}}
\end{figure}

Simulation results demonstrated that the proposed method significantly outperformed conventional approaches. The predicted output of the linearly identified model exhibited poor long-term predictive performance in ODE simulations. Diesel engine airpath systems exhibit complex nonlinear characteristics, yet linear models and Hankel DMD rely on linear realizations and therefore fail to deliver satisfactory results. We then examined the bilinear Koopman realization. First, for the bilinear Koopman model using polynomial basis functions--similar to previous studies~\cite{zhang_reduced-order_2023,wang_improved_2023}--the lifting functions were selected based on findings from SINDy modeling~\cite{yahagi_sparse_2025}. However, this approach proved ineffective for modeling diesel engine airpath systems. Since designers cannot successfully reuse the basis functions employed in SINDy, alternative approaches are required. Previous literature has used randomly centered RBFs as lifting functions, but as noted in~\cite{shi_deep_2022}, Koopman models with non-optimized RBFs tend to diverge. Our simulations confirmed this observation: while one-step predictions achieved very high R-squared exceeding 0.98, long-term predictions failed to meet performance requirements. Our proposed method addresses this issue by optimizing the centers of the RBFs used as lifting functions based on multi-step prediction performance. The bilinear Koopman model with RBF optimization demonstrated superior predictive accuracy compared to other methods. Furthermore, we confirmed that prediction and generalization performance depend strongly on the choice of arguments for the lifting function. Previous studies typically used embedded state variables consisting of time delays of inputs and outputs, but we found that these methods suffered from poor prediction and generalization. In contrast, our proposed argument selection strategy achieved significantly better performance. These results demonstrate the effectiveness of the proposed method.

\begin{figure}[t]%thb
% \begin{center}
\centering
\includegraphics[width=8.0cm]{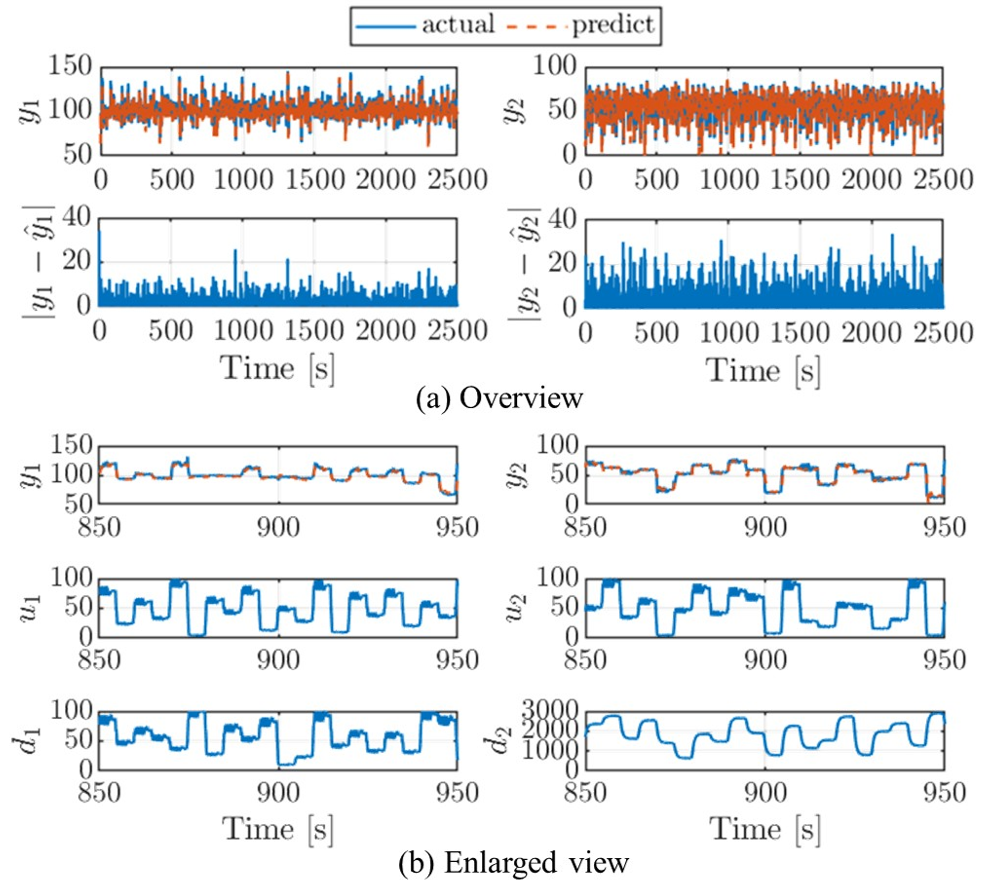}
\vspace{-1em} % 上に近づける
\caption{The long-term predictive performance of $y_1$ and $y_2$ for test data in simulation. 
\label{fig:SimDE_predict_valid_data}}
% \end{center}
\end{figure}

\section{Real-world experimental verification}
\label{sec:Experimental verification}

\subsection{Experimental setting}
\label{subsec:Experimental setting}

A real-world engine bench test, as shown in Fig.~\ref{fig:Eng_bench_test}, was conducted.
The experimental environment is the same as in the literature~\cite{ishizuka_model-free_2017}. 
The variable $y_2$ represents the mass air flow (MAF)~[mg/st], which is related to the EGR ratio, while the other variables are the same as those used in the simulation. In the engine bench tests, training and test data are collected under closed-loop conditions. The sampling period was 100 ms. In contrast, the simulation section used open-loop data, which are generally considered preferable. Here, we demonstrate that the proposed method remains effective in closed-loop experiments, where measurement noise tends to be amplified. The controller employed in the closed-loop experiments was a PID-based mass-production controller.

\begin{figure}[b]%thb
% \begin{center}
\centering
\includegraphics[width=3.5cm]{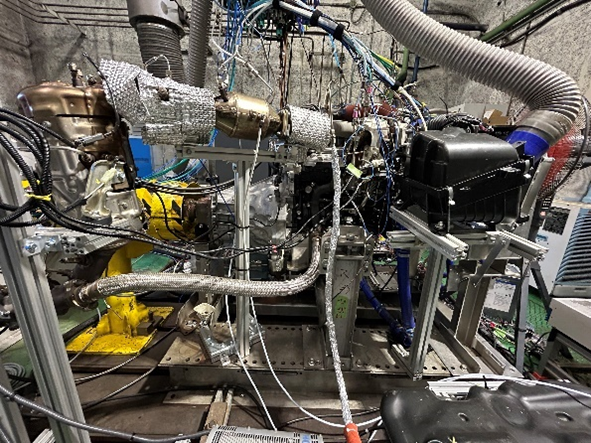}
\vspace{-0.5em} % 上に近づける
\caption{Engine bench test. Fuel type:diesel, Arrangement cylinders: 4 in line, Maximum power:110 kW (150 PS)\/3600 rpm, Maximum torque: 350 Nm (35.69 kgm)\/1800-2600 rpm (net), Combustion type: Direct injection with water-cooled 4-valve DOHC (double overhead camshaft).
\label{fig:Eng_bench_test}}
% \end{center}
\end{figure}

\subsection{Results and discussions}
\label{subsec:Experimental Results}
Fig.~\ref{fig:ExpDE_lrn_data} shows the  the learning and test data.
The Koopman models were derived from the learning data. 
Fig.~\ref{fig:ExpResults_LT_predict_valid} shows the prediction performance of the proposed method on the test data.
These signals are normalized. In addition, we compare the performance of the conventional and proposed methods. Table~\ref{tab:Exp_R2_long_term} reports the R-squared values for one-step predictions on the learning data and for both one-step and long-term predictions on the test data. The table also compares the results of the linear and bilinear Koopman realizations. The cases presented in Table~\ref{tab:Exp_R2_long_term} corresponds to those shown in Table~\ref{tab:simulation_cases}. From the results, the proposed method provides the best performance among them, as well as simulation verification. The LTI Koopman form is limited to the prediction performance. The predictions of the standard bilinear Koopman form, where lifting functions are \(\psi_x=\psi_w\) diverge for the test data. However, the proposed generalized bilinear Koopman form provides the desired performance for learning and test data. From the above, the effectiveness of the proposed methodology was experimentally verified.

\begin{table}[t]
\centering
\caption{R-squared of one-step and long-term predictions in the experimental test.}
\label{tab:Exp_R2_long_term}

\vspace{-1.0em} % 上に近づける
\begin{tabular}{c|c|c|c|c|c|c}
\hline
\multirow{3}{*}{\textbf{Case}} &
\multicolumn{2}{c|}{\makecell{\textbf{One-step}\\\textbf{prediction}}} &
\multicolumn{4}{c}{\textbf{Long-term prediction}} \\ \cline{2-7}
& \multicolumn{2}{c|}{learning data} &
\multicolumn{2}{c|}{learning data} &
\multicolumn{2}{c}{test data} \\ \cline{2-7}
& $y_1$ & $y_2$ & $y_1$ & $y_2$ & $y_1$ & $y_2$ \\ \hline
LK & 0.999 & 0.994 & 0.988 & 0.973 & 0.917 & 0.834 \\ \hline
BKL & 0.999 & 0.997 & 0.997 & 0.992 & $-0.774$ & $-3.68$ \\ \hline
GBKL  & \textbf{0.999} & \textbf{0.997} & \textbf{0.996} & \textbf{0.992} & \textbf{0.936} & \textbf{0.911} \\ \hline
\end{tabular}
\end{table}

\begin{figure}[t]%thb
% \begin{center}
\centering
\includegraphics[width=8.0cm]{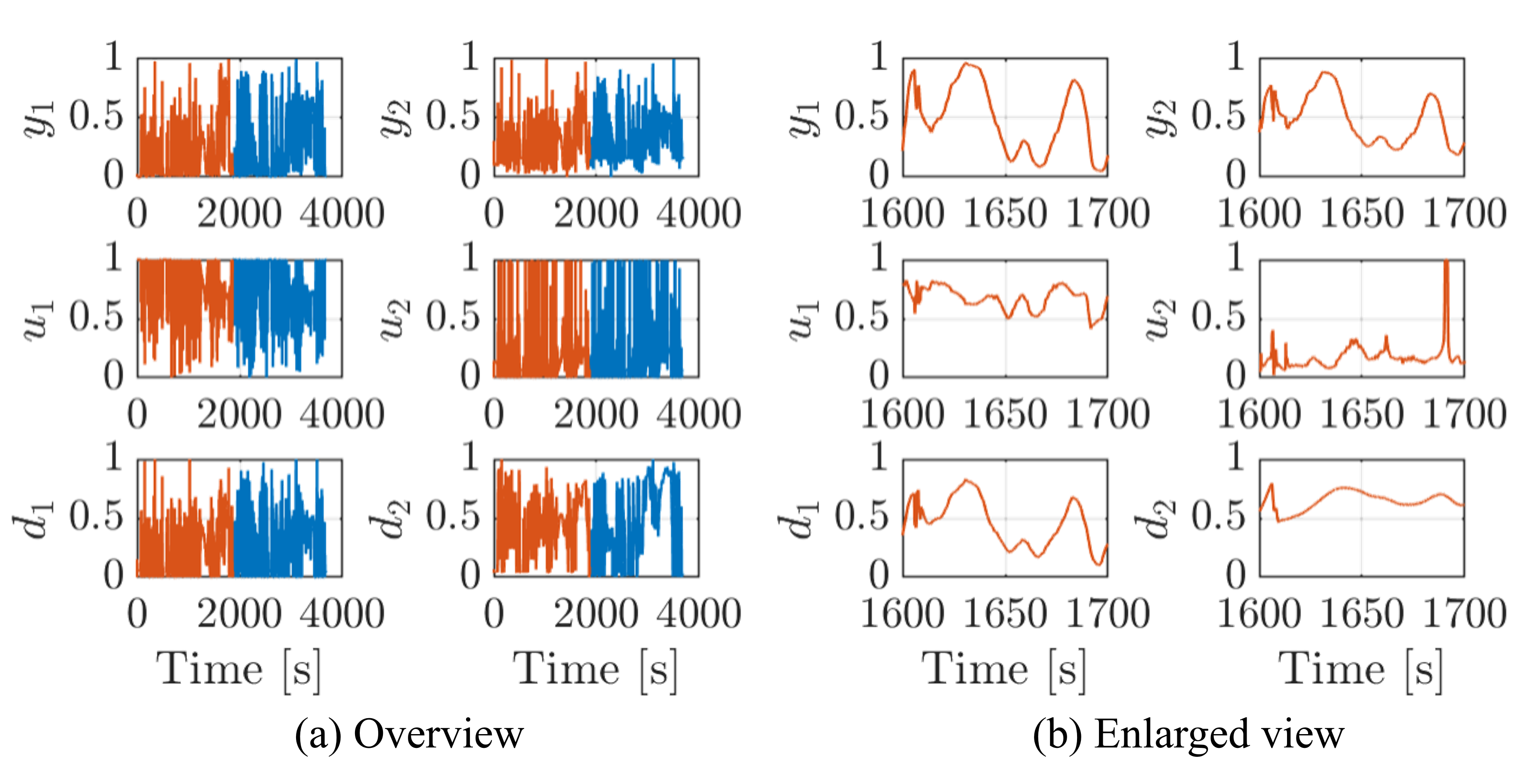}
\vspace{-1em} % 上に近づける
\caption{The learning and test data in experimental test. 
\label{fig:ExpDE_lrn_data}}
\end{figure}

\begin{figure}[t]%thb
    \centering
\includegraphics[width=8.0cm]{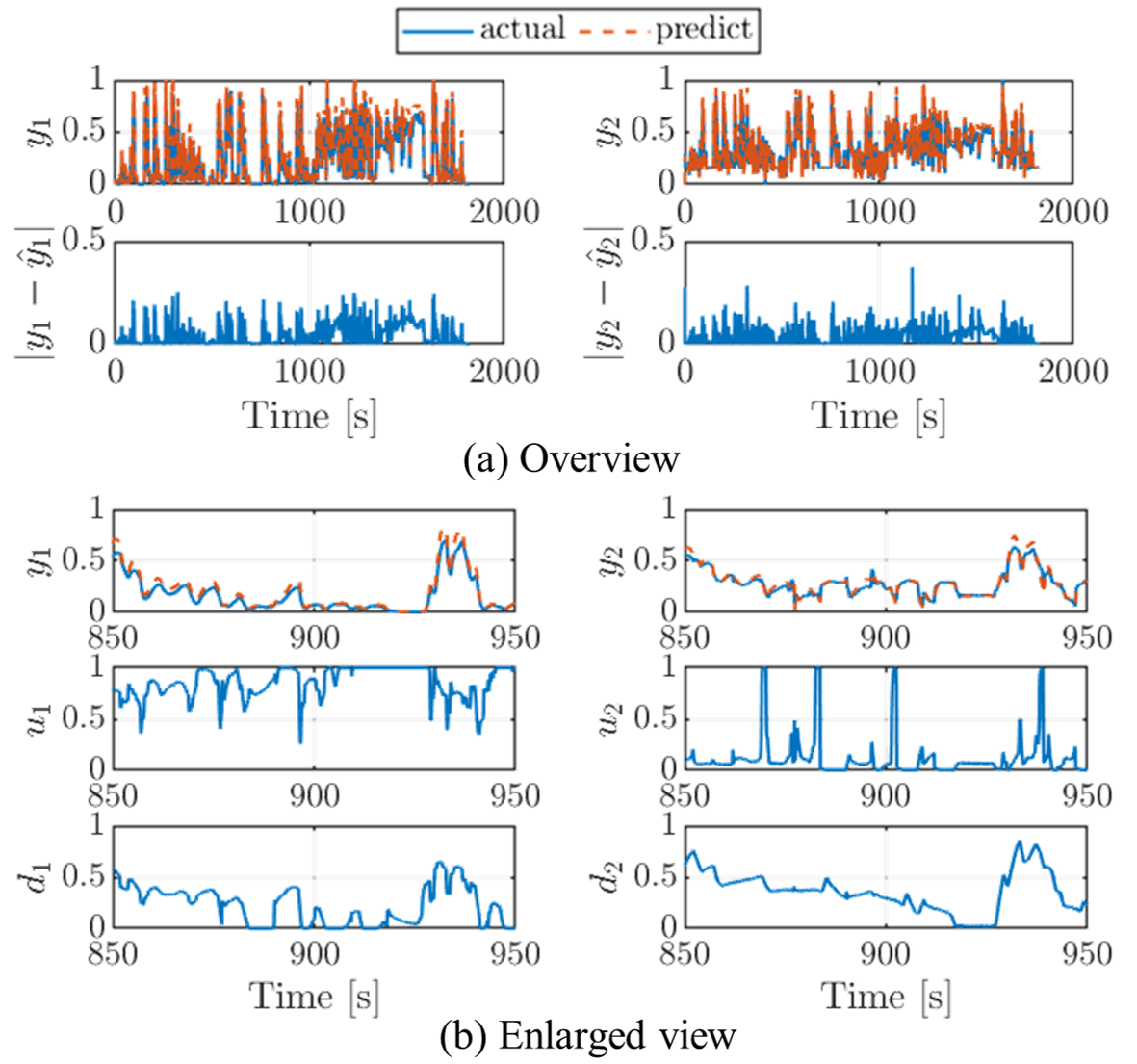}
\vspace{-1em} % 上に近づける
\caption{The long-term predictive performance of $y_1$ and $y_2$ for test data in experimental test. 
\label{fig:ExpResults_LT_predict_valid}}
% \end{center}
\end{figure}

Section~\ref{sec:Experimental verification} has examined the effectiveness of the proposed method through an actual engine bench test. Similar to the simulation verification, the proposed method--namely, the IO bilinear Koopman realization--achieves the best performance among the compared methods, including the LTI Koopman realization and the standard bilinear Koopman realization. In the simulation section, the learning data were acquired through an open-loop test. General system identification typically prefers rich excitation signals under open-loop testing conditions. In the experimental section, we assumed a more realistic scenario and obtained transient data from closed-loop experiments as training data. In other words, we considered that in actual tests, applying random vibrations using open-loop tests would not be easy and would require additional technical expertise. Even under such conditions, we demonstrated that the proposed method can achieve the desired results for both the training and test data. 

\section{Conclusions}
\label{sec:Conclusions}

This study presented an input-output bilinear Koopman realization with a lifting function optimization algorithm. The proposed framework integrates bilinear Koopman realization, input-output data utilization, and lifting function optimization. This approach enables Koopman-based modeling under practical constraints where only input-output data are available, significantly expanding its applicability to industrial systems. The key contributions of this work include the investigation of bilinear form structures for Koopman realization, the systematic selection of lifting function arguments when using input-output data, the optimization of lifting functions via PSO, and the successful application of the proposed method to an industrial diesel engine airpath system. Unlike previous studies that simply combined embedded states with input and output delays--thus including the time evolution of uninformative inputs--our method addresses argument selection to enhance predictive performance. Furthermore, because the design of the lifting function strongly influences accuracy and is difficult for users to determine, the proposed PSO-based optimization--where the centers of the RBFs employed as the lifting function are optimized--provides a practical and effective solution. The method was validated on a diesel engine intake and exhaust system exhibiting strong nonlinearity and MIMO characteristics. Simulation and experimental test results demonstrated substantial improvements in predictive accuracy compared to conventional approaches, achieving an R-squared over 0.90, which meets engineering requirements. Future work will focus on constructing controller with IO Koopman model.

%%%%%%%%%%%%%%%%%%%%%%%%%%%%%%%%%%%%%%%%%%%%%%%%%%%%%%%%%%%%%%%%%%%%%%%%%%%%%%%%

% References

\bibliographystyle{Bibliography/IEEEtran}
% \bibliography{Bibliography/IEEEabrv,Bibliography/BIB_xx-TIE-xxxx,Bibliography/Koopman,Bibliography/SINDy,Bibliography/DE,Bibliography/koopman_DE,Bibliography/Published,Bibliography/DMD}\ %IEEEabrv instead of IEEEfull
\bibliography{Bibliography/IEEEabrv,Bibliography/reference}\ %IEEEabrv 

@article{proctor_dynamic_2016,
	title = {Dynamic {Mode} {Decomposition} with {Control}},
	volume = {15},
	issn = {1536-0040},
	url = {https://epubs.siam.org/doi/10.1137/15M1013857},
	doi = {10.1137/15M1013857},
	language = {en},
	number = {1},
	urldate = {2026-02-10},
	journal = {SIAM Journal on Applied Dynamical Systems},
	author = {Proctor, Joshua L. and Brunton, Steven L. and Kutz, J. Nathan},
	month = jan,
	year = {2016},
	pages = {142--161},
}

@article{zhao_kalman-koopman_2024,
	title = {A {Kalman}-{Koopman} {LQR} {Control} {Approach} to {Robotic} {Systems}},
	volume = {71},
	copyright = {https://ieeexplore.ieee.org/Xplorehelp/downloads/license-information/IEEE.html},
	issn = {0278-0046, 1557-9948},
	url = {https://ieeexplore.ieee.org/document/10505835/},
	doi = {10.1109/TIE.2024.3379674},
	language = {en},
	number = {12},
	urldate = {2026-02-10},
	journal = {IEEE Transactions on Industrial Electronics},
	author = {Zhao, Dongdong and Yang, Xiaodi and Li, Yichang and Xu, Li and She, Jinhua and Yan, Shi},
	month = dec,
	year = {2024},
	pages = {16047--16056},
}

@article{meng_online_2025,
	title = {Online {Koopman} {Operator}-{Based} {Feedforward} {Compensation} {Strategy} for {Fast} {Tool} {Servos} {With} {Robust} {High}-{Bandwidth} {Control}},
	volume = {72},
	copyright = {https://ieeexplore.ieee.org/Xplorehelp/downloads/license-information/IEEE.html},
	issn = {0278-0046, 1557-9948},
	url = {https://ieeexplore.ieee.org/document/10634980/},
	doi = {10.1109/TIE.2024.3433486},
	language = {en},
	number = {3},
	urldate = {2026-02-09},
	journal = {IEEE Transactions on Industrial Electronics},
	author = {Meng, Yixuan and Li, Linlin and Wang, Xiangyuan and Huang, Wei-Wei and Wu, Edmond Q. and Zhu, LiMin},
	month = mar,
	year = {2025},
	pages = {2958--2967},
}

@article{bhattacharya_nonlinear_2022,
	title = {Nonlinear {Model} {Predictive} {Control} of a {Robotic} {Soft} {Esophagus}},
	volume = {69},
	copyright = {https://ieeexplore.ieee.org/Xplorehelp/downloads/license-information/IEEE.html},
	issn = {0278-0046, 1557-9948},
	url = {https://ieeexplore.ieee.org/document/9591327/},
	doi = {10.1109/TIE.2021.3121755},
	language = {en},
	number = {10},
	urldate = {2026-02-09},
	journal = {IEEE Transactions on Industrial Electronics},
	author = {Bhattacharya, Dipankar and Hashem, Ryman and Cheng, Leo K. and Xu, Weiliang},
	month = oct,
	year = {2022},
	pages = {10363--10373},
}

@article{yang_adaptive_2026,
	title = {Adaptive {Deep} {Koopman} {Operators} for {Soft} {Robot} {Control}: {Application} to {Luban} {Lock} {Disassembly}},
	copyright = {https://ieeexplore.ieee.org/Xplorehelp/downloads/license-information/IEEE.html},
	issn = {0278-0046, 1557-9948},
	shorttitle = {Adaptive {Deep} {Koopman} {Operators} for {Soft} {Robot} {Control}},
	url = {https://ieeexplore.ieee.org/document/11349691/},
	doi = {10.1109/TIE.2025.3634429},
	language = {en},
	urldate = {2026-02-09},
	journal = {IEEE Transactions on Industrial Electronics},
	author = {Yang, Yihe and Jiang, Yue and Cao, Wenyu and Li, Cong and Jiang, Wei and Mo, Hangjie and Li, Xiaojian and Xu, Xin and Zhang, Xinglong},
	year = {2026},
	pages = {1--12},
}

@article{huang_nitrogen_2023,
	title = {Nitrogen {Oxides} {Concentration} {Estimation} of {Diesel} {Engines} {Based} on a {Sparse} {Nonstationary} {Trigonometric} {Gaussian} {Process} {Regression} {With} {Maximizing} the {Composite} {Likelihood}},
	volume = {70},
	issn = {0278-0046, 1557-9948},
	url = {https://ieeexplore.ieee.org/document/10002403/},
	doi = {10.1109/TIE.2022.3231246},
	language = {en},
	number = {11},
	urldate = {2026-02-04},
	journal = {IEEE Transactions on Industrial Electronics},
	author = {Huang, Haojie and Peng, Xin and Du, Wei and Ding, Steven X. and Zhong, Weimin},
	month = nov,
	year = {2023},
	pages = {11744--11753},
}

@article{zhao_explicit_2014,
	title = {An {Explicit} {Model} {Predictive} {Control} {Framework} for {Turbocharged} {Diesel} {Engines}},
	volume = {61},
	issn = {0278-0046, 1557-9948},
	url = {http://ieeexplore.ieee.org/document/6584790/},
	doi = {10.1109/TIE.2013.2279353},
	language = {en},
	number = {7},
	urldate = {2026-02-04},
	journal = {IEEE Transactions on Industrial Electronics},
	author = {Zhao, Dezong and Liu, Cunjia and Stobart, Richard and Deng, Jiamei and Winward, Edward and Dong, Guangyu},
	month = jul,
	year = {2014},
	pages = {3540--3552},
}

@book{kaneko_model_2019,
	title = {Model {Based} {Control} for {Automotive} engines},
	isbn = {978-4-339-04661-8},
	language = {jpn},
	publisher = {CORONA},
	author = {Kaneko, Shigehiko and Yamasaki, Yudai and Ohmori, Hiromitsu and Mitsuo, Hirata and Mizumoto, Ikuro and Ichiyanagi, Mitsuhisa and Matsunaga, Akio and Jimbo, Tomohiko},
	year = {2019},
}

@article{haber_structure_1990,
	title = {Structure identification of nonlinear dynamic systems—{A} survey on input/output approaches},
	volume = {26},
	issn = {0005-1098},
	url = {https://www.sciencedirect.com/science/article/pii/000510989090044I},
	doi = {10.1016/0005-1098(90)90044-I},
	number = {4},
	urldate = {2024-04-11},
	journal = {Automatica},
	author = {Haber, R. and Unbehauen, H.},
	month = jul,
	year = {1990},
	pages = {651--677},
}

@article{van_overschee_n4sid_1994,
	series = {Special issue on statistical signal processing and control},
	title = {{N4SID}: {Subspace} algorithms for the identification of combined deterministic-stochastic systems},
	volume = {30},
	issn = {0005-1098},
	shorttitle = {{N4SID}},
	url = {https://www.sciencedirect.com/science/article/pii/0005109894902305},
	doi = {10.1016/0005-1098(94)90230-5},
	number = {1},
	urldate = {2023-12-08},
	journal = {Automatica},
	author = {Van Overschee, Peter and De Moor, Bart},
	month = jan,
	year = {1994},
	pages = {75--93},
}

@article{bevanda_koopman_2021,
	title = {Koopman {Operator} {Dynamical} {Models}: {Learning}, {Analysis} and {Control}},
	volume = {52},
	issn = {13675788},
	shorttitle = {Koopman {Operator} {Dynamical} {Models}},
	url = {http://arxiv.org/abs/2102.02522},
	doi = {10.1016/j.arcontrol.2021.09.002},
	urldate = {2023-07-18},
	journal = {Annual Reviews in Control},
	author = {Bevanda, Petar and Sosnowski, Stefan and Hirche, Sandra},
	year = {2021},
	pages = {197--212},
}

@misc{brunton_sparse_2016,
	title = {Sparse {Identification} of {Nonlinear} {Dynamics} with {Control} ({SINDYc})},
	url = {http://arxiv.org/abs/1605.06682},
	doi = {10.48550/arXiv.1605.06682},
	language = {en},
	urldate = {2025-06-02},
	publisher = {arXiv},
	author = {Brunton, Steven L. and Proctor, Joshua L. and Kutz, J. Nathan},
	month = may,
	year = {2016},
	note = {arXiv:1605.06682 [math]},
}

@article{brunton_modern_2022,
	title = {Modern {Koopman} {Theory} for {Dynamical} {Systems}},
	volume = {64},
	issn = {0036-1445, 1095-7200},
	url = {https://epubs.siam.org/doi/10.1137/21M1401243},
	doi = {10.1137/21M1401243},
	language = {en},
	number = {2},
	urldate = {2025-06-02},
	journal = {SIAM Review},
	author = {Brunton, Steven L. and Budišić, Marko and Kaiser, Eurika and Kutz, J. Nathan},
	month = may,
	year = {2022},
	pages = {229--340},
}

@inproceedings{goswami_global_2017,
	address = {Melbourne, Australia},
	title = {Global bilinearization and controllability of control-affine nonlinear systems: {A} {Koopman} spectral approach},
	isbn = {978-1-5090-2873-3},
	shorttitle = {Global bilinearization and controllability of control-affine nonlinear systems},
	url = {http://ieeexplore.ieee.org/document/8264582/},
	doi = {10.1109/CDC.2017.8264582},
	language = {en},
	urldate = {2023-12-22},
	booktitle = {2017 {IEEE} 56th {Annual} {Conference} on {Decision} and {Control} ({CDC})},
	publisher = {IEEE},
	author = {Goswami, Debdipta and Paley, Derek A.},
	month = dec,
	year = {2017},
	pages = {6107--6112},
}

@article{korda_linear_2018,
	title = {Linear predictors for nonlinear dynamical systems: {Koopman} operator meets model predictive control},
	volume = {93},
	issn = {00051098},
	shorttitle = {Linear predictors for nonlinear dynamical systems},
	url = {https://linkinghub.elsevier.com/retrieve/pii/S000510981830133X},
	doi = {10.1016/j.automatica.2018.03.046},
	language = {en},
	urldate = {2023-07-20},
	journal = {Automatica},
	author = {Korda, Milan and Mezić, Igor},
	month = jul,
	year = {2018},
	pages = {149--160},
}

@article{ishizuka_model-free_2017,
	title = {Model-free adaptive control scheme for {EGR}/{VNT} control of a diesel engine using the simultaneous perturbation stochastic approximation},
	volume = {39},
	issn = {0142-3312, 1477-0369},
	url = {http://journals.sagepub.com/doi/10.1177/0142331215602327},
	doi = {10.1177/0142331215602327},
	language = {en},
	number = {1},
	urldate = {2023-08-07},
	journal = {Transactions of the Institute of Measurement and Control},
	author = {Ishizuka, Shinichi and Kajiwara, Itsuro and Sato, Junichi and Hanamura, Yoshifumi and Hanawa, Satoshi},
	month = jan,
	year = {2017},
	pages = {114--128},
}

@article{lusch_deep_2018,
	title = {Deep learning for universal linear embeddings of nonlinear dynamics},
	volume = {9},
	issn = {2041-1723},
	url = {https://www.nature.com/articles/s41467-018-07210-0},
	doi = {10.1038/s41467-018-07210-0},
	language = {en},
	number = {1},
	urldate = {2024-11-20},
	journal = {Nature Communications},
	author = {Lusch, Bethany and Kutz, J. Nathan and Brunton, Steven L.},
	month = nov,
	year = {2018},
	pages = {4950},
}

@article{moriyasu_diesel_2019,
	title = {Diesel engine air path control based on neural approximation of nonlinear {MPC}},
	volume = {91},
	issn = {09670661},
	url = {https://linkinghub.elsevier.com/retrieve/pii/S0967066119301303},
	doi = {10.1016/j.conengprac.2019.104114},
	language = {en},
	urldate = {2023-08-04},
	journal = {Control Engineering Practice},
	author = {Moriyasu, Ryuta and Nojiri, Sayaka and Matsunaga, Akio and Nakamura, Toshihiro and Jimbo, Tomohiko},
	month = oct,
	year = {2019},
	pages = {104114},
}

@article{wang_improved_2023,
	title = {An {Improved} {Koopman}-{MPC} {Framework} for {Data}-{Driven} {Modeling} and {Control} of {Soft} {Actuators}},
	volume = {8},
	issn = {2377-3766},
	doi = {10.1109/LRA.2022.3229235},
	number = {2},
	journal = {IEEE Robotics and Automation Letters},
	author = {Wang, Jiajin and Xu, Baoguo and Lai, Jianwei and Wang, Yifei and Hu, Cong and Li, Huijun and Song, Aiguo},
	month = feb,
	year = {2023},
	pages = {616--623},
}

@article{williams_extending_2016,
	title = {Extending {Data}-{Driven} {Koopman} {Analysis} to {Actuated} {Systems}},
	volume = {49},
	issn = {24058963},
	url = {https://linkinghub.elsevier.com/retrieve/pii/S2405896316318286},
	doi = {10.1016/j.ifacol.2016.10.248},
	language = {en},
	number = {18},
	urldate = {2023-08-02},
	journal = {IFAC-PapersOnLine},
	author = {Williams, Matthew O. and Hemati, Maziar S. and Dawson, Scott T.M. and Kevrekidis, Ioannis G. and Rowley, Clarence W.},
	year = {2016},
	pages = {704--709},
}

@article{schaible_fuzzy_1997,
	title = {Fuzzy logic models for ranking process effects},
	volume = {5},
	copyright = {https://ieeexplore.ieee.org/Xplorehelp/downloads/license-information/IEEE.html},
	issn = {10636706},
	url = {http://ieeexplore.ieee.org/document/649905/},
	doi = {10.1109/91.649905},
	language = {en},
	number = {4},
	urldate = {2024-04-11},
	journal = {IEEE Transactions on Fuzzy Systems},
	author = {Schaible, B. and {Hong Xie} and {Yung-Cheng Lee}},
	month = nov,
	year = {1997},
	pages = {545--556},
}

@article{o_williams_kernel-based_2015,
	title = {A kernel-based method for data-driven koopman spectral analysis},
	volume = {2},
	issn = {2158-2505},
	url = {http://aimsciences.org//article/doi/10.3934/jcd.2015005},
	doi = {10.3934/jcd.2015005},
	language = {en},
	number = {2},
	urldate = {2024-11-21},
	journal = {Journal of Computational Dynamics},
	author = {O. Williams, Matthew and W. Rowley, Clarence and G.  Kevrekidis, Ioannis},
	year = {2015},
	pages = {247--265},
}

@misc{son_handling_2020,
	title = {Handling plant-model mismatch in {Koopman} {Lyapunov}-based model predictive control via offset-free control framework},
	url = {http://arxiv.org/abs/2010.07239},
	doi = {10.48550/arXiv.2010.07239},
	urldate = {2023-07-31},
	publisher = {arXiv},
	author = {Son, Sang Hwan and Narasingam, Abhinav and Kwon, Joseph Sang-Il},
	month = oct,
	year = {2020},
	note = {arXiv:2010.07239 [cs, eess]},
}

@article{schmid_dynamic_2010,
	title = {Dynamic mode decomposition of numerical and experimental data},
	volume = {656},
	issn = {0022-1120, 1469-7645},
	url = {https://www.cambridge.org/core/product/identifier/S0022112010001217/type/journal_article},
	doi = {10.1017/S0022112010001217},
	language = {en},
	urldate = {2023-07-31},
	journal = {Journal of Fluid Mechanics},
	author = {Schmid, Peter J.},
	month = aug,
	year = {2010},
	pages = {5--28},
}

@article{proctor_generalizing_2018,
	title = {Generalizing {Koopman} {Theory} to {Allow} for {Inputs} and {Control}},
	volume = {17},
	url = {https://epubs.siam.org/doi/10.1137/16M1062296},
	doi = {10.1137/16M1062296},
	number = {1},
	urldate = {2023-08-02},
	journal = {SIAM Journal on Applied Dynamical Systems},
	author = {Proctor, Joshua L. and Brunton, Steven L. and Kutz, J. Nathan},
	month = jan,
	year = {2018},
	publisher = {{Society for Industrial and Applied Mathematics}},
	pages = {909--930},
}

@article{xiao_deep_2023,
	title = {Deep {Neural} {Networks} {With} {Koopman} {Operators} for {Modeling} and {Control} of {Autonomous} {Vehicles}},
	volume = {8},
	issn = {2379-8904},
	doi = {10.1109/TIV.2022.3180337},
	number = {1},
	journal = {IEEE Transactions on Intelligent Vehicles},
	author = {Xiao, Yongqian and Zhang, Xinglong and Xu, Xin and Liu, Xueqing and Liu, Jiahang},
	month = jan,
	year = {2023},
	pages = {135--146},
}

@article{yahagi_sparse_2025,
	title = {Sparse {Identification} and {Nonlinear} {Model} {Predictive} {Control} for {Diesel} {Engine} {Air} {Path} {System}},
	volume = {23},
	issn = {2005-4092},
	url = {https://doi.org/10.1007/s12555-024-0452-9},
	doi = {10.1007/s12555-024-0452-9},
	language = {en},
	number = {2},
	urldate = {2025-04-22},
	journal = {International Journal of Control, Automation and Systems},
	author = {Yahagi, Shuichi and Seto, Hiroki and Yonezawa, Ansei and Kajiwara, Itsuro},
	month = feb,
	year = {2025},
	pages = {620--629},
}

@article{zhang_reduced-order_2023,
	title = {Reduced-order {Koopman} modeling and predictive control of nonlinear processes},
	volume = {179},
	issn = {00981354},
	url = {https://linkinghub.elsevier.com/retrieve/pii/S0098135423003101},
	doi = {10.1016/j.compchemeng.2023.108440},
	language = {en},
	urldate = {2023-12-07},
	journal = {Computers \& Chemical Engineering},
	author = {Zhang, Xuewen and Han, Minghao and Yin, Xunyuan},
	month = nov,
	year = {2023},
	pages = {108440},
}

@article{huang_lstm-mpc_2023,
	title = {{LSTM}-{MPC}: {A} {Deep} {Learning} {Based} {Predictive} {Control} {Method} for {Multimode} {Process} {Control}},
	volume = {70},
	copyright = {https://ieeexplore.ieee.org/Xplorehelp/downloads/license-information/IEEE.html},
	issn = {0278-0046, 1557-9948},
	shorttitle = {{LSTM}-{MPC}},
	url = {https://ieeexplore.ieee.org/document/9994755/},
	doi = {10.1109/tie.2022.3229323},
	language = {en},
	number = {11},
	urldate = {2025-07-09},
	journal = {IEEE Transactions on Industrial Electronics},
	author = {Huang, Keke and Wei, Ke and Li, Fanbiao and Yang, Chunhua and Gui, Weihua},
	month = nov,
	year = {2023},
	Publisher = {Institute of Electrical and Electronics Engineers (IEEE)},
	pages = {11544--11554},
}

@article{zhang_bayesian_2024,
	title = {A {Bayesian} transfer sparse identification method for nonlinear {ARX} systems},
	volume = {38},
	copyright = {http://onlinelibrary.wiley.com/termsAndConditions\#vor},
	issn = {0890-6327, 1099-1115},
	url = {https://onlinelibrary.wiley.com/doi/10.1002/acs.3884},
	doi = {10.1002/acs.3884},
	language = {en},
	number = {10},
	urldate = {2025-07-22},
	journal = {International Journal of Adaptive Control and Signal Processing},
	author = {Zhang, Kang and Luan, Xiaoli and Ding, Feng and Liu, Fei},
	month = oct,
	year = {2024},
	note = {Publisher: Wiley},
	pages = {3484--3502},
}

@misc{abtahi_deep_2025,
	title = {Deep {Bilinear} {Koopman} {Model} for {Real}-{Time} {Vehicle} {Control} in {Frenet} {Frame}},
	url = {http://arxiv.org/abs/2507.12578},
	doi = {10.48550/arXiv.2507.12578},
	language = {en},
	urldate = {2025-10-24},
	publisher = {arXiv},
	author = {Abtahi, Mohammad and Araghi, Farhang Motallebi and Mojahed, Navid and Nazari, Shima},
	month = jul,
	year = {2025},
	note = {arXiv:2507.12578 [eess]},
}

@article{wang_deep_2024,
	title = {Deep bilinear {Koopman} realization for dynamics modeling and predictive control},
	volume = {15},
	issn = {1868-8071, 1868-808X},
	url = {https://link.springer.com/10.1007/s13042-023-02095-y},
	doi = {10.1007/s13042-023-02095-y},
	language = {en},
	number = {8},
	urldate = {2025-10-26},
	journal = {International Journal of Machine Learning and Cybernetics},
	author = {Wang, Meixi and Lou, Xuyang and Cui, Baotong},
	month = aug,
	year = {2024},
	pages = {3327--3352},
}

@article{bruder_advantages_2021,
	title = {Advantages of {Bilinear} {Koopman} {Realizations} for the {Modeling} and {Control} of {Systems} {With} {Unknown} {Dynamics}},
	volume = {6},
	copyright = {https://ieeexplore.ieee.org/Xplorehelp/downloads/license-information/IEEE.html},
	issn = {2377-3766, 2377-3774},
	url = {https://ieeexplore.ieee.org/document/9384174/},
	doi = {10.1109/LRA.2021.3068117},
	language = {en},
	number = {3},
	urldate = {2025-10-26},
	journal = {IEEE Robotics and Automation Letters},
	author = {Bruder, Daniel and Fu, Xun and Vasudevan, Ram},
	month = jul,
	year = {2021},
	pages = {4369--4376},
}

@article{yu_autonomous_2022,
	title = {Autonomous {Driving} using {Linear} {Model} {Predictive} {Control} with a {Koopman} {Operator} based {Bilinear} {Vehicle} {Model}},
	volume = {55},
	issn = {24058963},
	url = {https://linkinghub.elsevier.com/retrieve/pii/S2405896322023254},
	doi = {10.1016/j.ifacol.2022.10.293},
	language = {en},
	number = {24},
	urldate = {2025-10-26},
	journal = {IFAC-PapersOnLine},
	author = {Yu, Siyuan and Shen, Congkai and Ersal, Tulga},
	year = {2022},
	pages = {254--259},
}

@article{williams_datadriven_2015,
	title = {A {Data}–{Driven} {Approximation} of the {Koopman} {Operator}: {Extending} {Dynamic} {Mode} {Decomposition}},
	volume = {25},
	issn = {0938-8974, 1432-1467},
	shorttitle = {A {Data}–{Driven} {Approximation} of the {Koopman} {Operator}},
	url = {http://link.springer.com/10.1007/s00332-015-9258-5},
	doi = {10.1007/s00332-015-9258-5},
	language = {en},
	number = {6},
	urldate = {2025-10-28},
	journal = {Journal of Nonlinear Science},
	author = {Williams, Matthew O. and Kevrekidis, Ioannis G. and Rowley, Clarence W.},
	month = dec,
	year = {2015},
	pages = {1307--1346},
}

@article{shi_deep_2022,
	title = {Deep {Koopman} {Operator} {With} {Control} for {Nonlinear} {Systems}},
	volume = {7},
	copyright = {https://ieeexplore.ieee.org/Xplorehelp/downloads/license-information/IEEE.html},
	issn = {2377-3766, 2377-3774},
	url = {https://ieeexplore.ieee.org/document/9799788/},
	doi = {10.1109/LRA.2022.3184036},
	language = {en},
	number = {3},
	urldate = {2025-11-21},
	journal = {IEEE Robotics and Automation Letters},
	author = {Shi, Haojie and Meng, Max Q.-H.},
	month = jul,
	year = {2022},
	pages = {7700--7707},
}

@article{zhao_deep_2024,
	title = {Deep {Bilinear} {Koopman} {Model} {Predictive} {Control} for {Nonlinear} {Dynamical} {Systems}},
	volume = {71},
	copyright = {https://ieeexplore.ieee.org/Xplorehelp/downloads/license-information/IEEE.html},
	issn = {0278-0046, 1557-9948},
	url = {https://ieeexplore.ieee.org/document/10527398/},
	doi = {10.1109/TIE.2024.3390717},
	language = {en},
	number = {12},
	urldate = {2025-11-21},
	journal = {IEEE Transactions on Industrial Electronics},
	author = {Zhao, Dongdong and Li, Boyu and Lu, Fuxiang and She, Jinhua and Yan, Shi},
	month = dec,
	year = {2024},
	pages = {16077--16086},
}

@article{xiawen_parameters_2021,
	title = {Parameters determination of time-delayed embedding with application to {Koopman} operator-based model predictive frequency control},
	issn = {20960042},
	url = {https://ieeexplore.ieee.org/stamp/stamp.jsp?tp=&arnumber=9535413},
	doi = {10.17775/CSEEJPES.2021.02000},
	language = {en},
	urldate = {2025-11-25},
	journal = {CSEE Journal of Power and Energy Systems},
	author = {Xiawen, Li and Chetan, Mishra and Shichao, Chen and Yajun, Wang and Jaime, De La Ree},
	year = {2021},
}

@article{iwadare_multi-variable_2009,
	title = {Multi-{Variable} {Air}-{Path} {Management} for a {Clean} {Diesel} {Engine} {Using} {Model} {Predictive} {Control}},
	volume = {02},
	issn = {1946-3936, 1946-3944},
	url = {https://saemobilus.sae.org/articles/multi-variable-air-path-management-a-clean-diesel-engine-using-model-predictive-control-2009-01-0733},
	doi = {10.4271/2009-01-0733},
	language = {en},
	number = {1},
	urldate = {2025-12-02},
	journal = {SAE International Journal of Engines},
	author = {Iwadare, Mitsuhiro and Ueno, Masaki and Adachi, Shuichi},
	month = apr,
	year = {2009},
	pages = {764--773},
}

@article{gagliardi_direct_2014,
	title = {Direct {C}/{GMRES} {Control} of {The} {Air} {Path} of a {Diesel} {Engine}},
	volume = {47},
	issn = {14746670},
	url = {https://linkinghub.elsevier.com/retrieve/pii/S1474667016420677},
	doi = {10.3182/20140824-6-ZA-1003.02481},
	language = {en},
	number = {3},
	urldate = {2025-12-02},
	journal = {IFAC Proceedings Volumes},
	author = {Gagliardi, Davide and Othsuka, T. and Re, L. Del},
	year = {2014},
	pages = {3000--3005},
}

@article{wei_gain_2007,
	title = {Gain {Scheduled} \${H}\_\{{\textbackslash}infty\}\$ {Control} for {Air} {Path} {Systems} of {Diesel} {Engines} {Using} {LPV} {Techniques}},
	volume = {15},
	copyright = {https://ieeexplore.ieee.org/Xplorehelp/downloads/license-information/IEEE.html},
	issn = {1063-6536},
	url = {http://ieeexplore.ieee.org/document/4162497/},
	doi = {10.1109/TCST.2007.894633},
	language = {en},
	number = {3},
	urldate = {2025-12-02},
	journal = {IEEE Transactions on Control Systems Technology},
	author = {Wei, Xiukun and Del Re, Luigi},
	month = may,
	year = {2007},
	pages = {406--415},
}

@article{shi_koopman_2023,
	title = {Koopman {Operators} for {Modeling} and {Control} of {Soft} {Robotics}},
	volume = {4},
	issn = {2662-4087},
	url = {https://link.springer.com/10.1007/s43154-023-00099-8},
	doi = {10.1007/s43154-023-00099-8},
	language = {en},
	number = {2},
	urldate = {2025-12-01},
	journal = {Current Robotics Reports},
	author = {Shi, Lu and Liu, Zhichao and Karydis, Konstantinos},
	month = jul,
	year = {2023},
	pages = {23--31},
}

@article{zawacki_dynamic_2023,
	title = {Dynamic {Mode} {Decomposition} for {Control} {Systems} with {Input} {Delays}},
	volume = {56},
	issn = {24058963},
	url = {https://linkinghub.elsevier.com/retrieve/pii/S2405896323023406},
	doi = {10.1016/j.ifacol.2023.12.007},
	language = {en},
	number = {3},
	urldate = {2025-12-10},
	journal = {IFAC-PapersOnLine},
	author = {Zawacki, Christopher C. and Abed, Eyad H.},
	year = {2023},
	pages = {97--102},
}

@article{yonezawa_sparse_2026,
	title = {Sparse {Identification} of {Nonlinear} {Dynamics} {With} {Library} {Optimization} {Mechanism}: {Recursive} {Long}-{Term} {Prediction} {Perspective}},
	language = {en},
	journal = {IEEE Transactions on Cybernetics},
	author = {Yonezawa, Ansei and Yonezawa, Heisei and Yahagi, Shuichi and Kajiwara, Itsuro and Kijimoto, Shinya and Taniuchi, Hikaru and Murakami, Kentaro},
	year = {2026},
}

\end{document}